\documentclass[letter]{IEEEtran}
\usepackage{algorithm,algpseudocode,graphicx}  
\usepackage{amsfonts}
\usepackage{tabularx}
\usepackage{subfigure}
\usepackage{amssymb}
\usepackage{amsthm}
\usepackage{booktabs} 
\usepackage{multirow} 
\usepackage{epsfig}
\usepackage{epstopdf} 
\usepackage{caption}

\usepackage[noadjust]{cite}
\newtheorem{theorem}{Theorem}
\newtheorem{remark}{Remark}

\newtheorem{proposition}{Proposition}
\newtheorem{lemma}{Lemma}

\usepackage[cmex10]{amsmath}
\allowdisplaybreaks  
   \usepackage{color}
 
\usepackage{subfigure}
\usepackage[bottom]{footmisc}
\usepackage{url}

\begin{document}
  \title{ QoS-aware Stochastic Spatial PLS Model for Analysing  Secrecy Performance under Eavesdropping and Jamming  } 
%\title{Title}
\author 
  {Bhawna Ahuja,~\IEEEmembership{Student~Member,~IEEE,}  Deepak Mishra,~\IEEEmembership{Member,~IEEE,}  and Ranjan Bose,~\IEEEmembership{Senior~Member,~IEEE} 
  \thanks{B. Ahuja  and R. Bose are  with the Bharti School of  Telecommunication Technology and Management, IIT Delhi,110016 New Delhi.  R. Bose is also affiliated with IIIT Delhi.   e-mail: (bhawna.ahuja,rbose)@iitd.ac.in}
  \thanks{D. Mishra is with the Department of
  School of Electrical Engineering and Telecommunications, UNSW Sydney, NSW 2052, Australia
   (e-mail: d.mishra@unsw.edu.au).} 
    \thanks{This work has been supported by the Department of Electron and Inform. Technol., Govt. of India under Vishvesvaraya PhD Fellowship scheme.}
%  % %\thanks{Digital Object Identifier xxxxxxxxxxxxxxxxx}
}
\maketitle
\begin{abstract}
    Securing wireless communication, being inherently vulnerable to eavesdropping and jamming attacks, becomes more challenging in resource-constrained networks like Internet-of-Things. Towards this, physical layer security (PLS) has gained significant attention due to its low complexity. In this paper, we address the issue of random inter-node distances in secrecy analysis and develop a comprehensive quality-of-service (QoS) aware PLS framework for the analysis of both eavesdropping and jamming  capabilities of attacker. The proposed solution covers spatially stochastic deployment of legitimate nodes and  attacker. We characterise the secrecy outage performance against both attacks using inter-node distance based probabilistic distribution functions. The model takes into account the practical limits arising out of underlying QoS requirements, which include the maximum distance between legitimate users driven by transmit power and receiver sensitivity. A novel concept of eavesdropping zone is introduced, and relative impact of jamming power is investigated. Closed-form expressions for asymptotic secrecy outage probability are derived offering insights into design of optimal system parameters for desired security level against the attacker's capability of both attacks.  Analytical framework, validated by numerical results, establishes that the proposed solution offers potentially accurate characterisation of the PLS performance and key design perspective from point-of-view of both legitimate user and  attacker.
\end{abstract} 

% \begin{keywords}
% Physical Layer Security, eavesdropping, jamming, secrecy rate, secrecy outage, distance distribution, stochastic analysis, Internet-of-Things, communication range.% %  %QoS-aware, secrecy outage analysis, high SNR regime, maximum separation between nodes, random inter-node separation, attacking capability,  eavesdropping range, jamming power, ratio distribution of SNRs.
% \end{keywords}  
\maketitle
\section{Introduction}\label{sec:intro}
\let\thefootnote\relax\footnote{A preliminary version of this paper has been  presented at IEEE ICC in  Shanghai, China, May 2019\cite{myICC}.}
Wireless communication owing to its open and broadcast nature is highly vulnerable to eavesdropping and jamming attacks. Recently, physical layer security (PLS) has recently drawn remarkable attention of the researchers in  resources constrained wireless networks like Internet of Things (IoT)  \cite{IoTPLS15}.
%he problem becomes extremely severe in next-generation wireless networks and information infrastructure like Internet of Things (IoT) \cite{IoT18}. 
%Considering constrained resources in these systems in terms of computing power and memory space, physical layer security (PLS) has recently drawn remarkable attention of the researchers  due to low computational requirements \cite{IoTPLS15,PLS-BookCh}.
 {Various PLS techniques have been reported in the literature with a major focus on eavesdropping attack 
 and friendly jamming. Here, legitimate nodes leverage the interference to secure the communication and 
 %where 
 the attacker is restricted to act as a mere passive adversary \cite{E1_TWC, E1_TVT, R1_TcomJam, R1_IET, R1_TWC2PLsurvey}.} However, the attacker may also use the interference to its advantage by jamming the legitimate reception other than eavesdropping. It is, therefore, of interest to study the attacker in both  eavesdropping and jamming modes. 
Furthermore,  spatial configurations of legitimate nodes and  attackers are, in general, modelled in deterministic manner. This assumption is only valid  if the location of the node is known. 
%In general scenarios, inter-node distances may be random with the availability of only stochastic information  about node's location.
Therefore, appropriate  modelling of a stochastic network is needed  to investigate the critical role of random inter-node separations on secrecy performance when the exact location information is not available. 
\subsection{Related Art}\vspace{-1mm}
 PLS is extensively explored to analyse secrecy capacity (SC) and secrecy outage probability (SOP) in various scenarios including High SNR regime \cite{ Liu2013, R32_noma,IET_R21, IET_R22}.
 %to enhance security by exploiting the unpredictable features of wireless channels.
 %It requires relatively low implementation complexity and energy consumption while it is capable of providing a strict level of security  \cite{IoTPLS15,survey18,Deepak}. n
  {In \cite{R32_noma}, artificial noise is exploited to enhance security in
 multiple-input single-output non-orthogonal multiple access (MISO-NOMA) systems by developing a secrecy beamforming scheme. 
%  Using the wiretap model \cite {Wyner} for PLS, secrecy capacity (SC) and secrecy outage probability (SOP) are investigated in many scenarios by taking fading and quasi-static fading  into consideration \cite{ Bloch, Bose }.
% % The presence of colluding eavesdroppers is considered in \cite{Goel}, but restricting its attention to the eavesdroppers placed at the same spatial location.
% High Signal to Noise Ratio (SNR) regime  is also explored to investigate the behaviour of SC and SOP over correlated ergodic fading channels \cite{Jeon2011,Liu2013}.
Impact of imperfect channel state information (CSI) is also widely explored in analysing the network performance \cite{IET_R21, IET_R22} when perfect CSI is not available.}
% imperfect CSI  is extensively used in wireless communication systems when perfect CSI is not available.
However, these works  ignored the randomness caused by large-scale propagation losses.
%Recently, secure communication is explored in stochastic network by assuming
%spatial distribution of the nodes as a homogeneous  Poisson Point Process (PPP) or Binomial Point Process (BPP) in \cite{Pinto08,ICC18,BPP}.
%PLS has gained considerable attention in spatial stochastic networks under the assumption of homogeneous PPP based deployment of nodes \cite{Pinto08,Pinto1,distancefading2016, ICC18, BPP}.
There is another line of research exploring the secure communication in spatial stochastic networks  \cite{Pinto08,Pinto1,distancefading2016,secrecyguard18}.   
%stochastic wireless networks \cite{Pinto08,Pinto1}. 
%In \cite{distancefading2016},  secrecy performance is analysed in stochastic network with a secrecy protected and interferer protect zone respectively surrounding the transmitter and  the legitimate receivers. Concept of  secrecy guard zone is explored in \cite{secrecyguard18}.  
These works explore the secrecy performance under a passive eavesdropper assuming the spatial distribution of the nodes as a homogeneous  Poisson Point Process (PPP) which may be  a good approximation for large-scale networks with known network density. 
%when the network density is known. 
{The stationary and isotropic  properties of homogeneous PPP consider that characteristics of network   as viewed from a node's aspect are similar for all the  nodes. However,  this assumption is not valid in  practical networks, especially having a finite number of nodes within a given area \cite{finite}}. 
%Few works \cite{finite,BPP} assume Binomial Point Process (BPP) based deployment which can take care of a network with finite  number of nodes but
Further, it is incapable of  analysing  the average performance measures at a randomly deployed node which has significant importance  in realistic networks. Some scenarios showing limitations of above models under   Device-to-Device (D2D) or sensor networks  has been examined in \cite{nodaldistance,D2D}. %Secrecy  performance  for Device-to Device (D2D) networks  is analyzed over alpha-$\mu$ fading channels with randomly distributed eavesdroppers under the assumption of fixed distance between source and  user \cite{ICC18}.  

  Apart from aforementioned works,  studies in %\cite {Gungor,game18, Ryu,Chen,hybrid18}
  \cite{Wyner_Encoding,Garnaev,MyWCL} exhibit a recent research interest on hybrid attackers that can either eavesdrop or jam. 
%These are also known as active eavesdropper or opportunistic attacker in literature.
%In \cite{Gungor}, SC  of block fading channel is analyzed in presence of  hybrid attacker with an arbitrary strategy to jam or eavesdrop.  
%The works of \cite{Poor11,Garnaev}  considered a scenario wherein a hybrid attacker  attempts to decrease the total network throughput by deciding   to  jam or eavesdrop so that it generates the most adverse circumstances for secure communication. %Security mechanisms are developed against hybrid attacks in \cite{Ryu,Chen,hybrid18}.
{The authors have investigated the full duplex attacker  in \cite{R31_noma,MyICC20} where  passive eavesdropping and active jamming are performed simultaneously.}
% The authors have investigated the joint impact on eavesdropper and jammer in  \cite{Rawat}  by considering SOP, average SC of large-scale multiple input multiple output (MIMO) wireless system. 
In \cite{TWC19Hybrid} secrecy performance of a wireless network with randomly deployed hybrid attackers is analysed  using stochastic geometry tools and random matrix theory.   Recently, much attention is being paid to  the utilisation  of distance distributions as a complement to PPP models for performance analysis in wireless networks \cite{nodaldistance,D2D} without considering secrecy requirements. 
\vspace{-2mm}
\subsection{Research Gap and Motivation}
\vspace{-1mm}
As noted in the literature survey, PLS  in spatial stochastic networks is mainly studied under the assumption of homogeneous PPP  based deployment of nodes \cite{Pinto08,Pinto1,distancefading2016, secrecyguard18,TWC19Hybrid}
%As noted in the literature survey, PLS has gained considerable attention in spatial stochastic networks under the assumption of homogeneous PPP  based deployment of nodes 
  which are non-viable in several  practical scenarios \cite{nodaldistance,D2D}.  
 Moreover, most  studies of stochastic networks consider a passive attacker that can only eavesdrop. Recently, works in \cite{Wyner_Encoding,Garnaev,MyWCL} have  shown the growing interest in hybrid attackers. But, above-mentioned works  consider the deterministic path-loss. The authors in  \cite{TWC19Hybrid}  have  studied the hybrid  attacker considering random path-loss,  however, analysis is done under the PPP assumption and available perfect CSI of the users to the source.    These observations reveal that  secrecy analysis  in a spatial stochastic network for  a general scenario  under eavesdropping as well as jamming mode of the attacker is still an open research problem. 
   In practical networks including IoT and D2D, legitimate nodes  as well as attacker  may be deployed randomly in the deployment regions; consequently, the exact location of any node may not be  available. %Since SC is a function of legitimate link distance as well as attacker link distance,  it is imperative to do accurate characterization of security performance to  assess the important role of random   inter-node separations  when both the link distances are random. 
    Furthermore, connection between two nodes will be set up  when  the geographical separation between them is smaller than a predefined threshold for retaining Quality-of-Service (QoS). In a similar way, attacker's capabilities to  eavesdrop and jam  also influence the secrecy performance;  hence they are indirect measure of QoS attributes from attacker's point of view. Eavesdropping capability is practically constrained by hardware limitations while jamming capability is constrained by power consumption. It is to be mentioned here that these QoS-governing  parameters did not get due attention in literature.  These facts motivate the study of this work for developing a practical and comprehensive stochastic model  so as to quantify secrecy performance in a realistic manner.  
    % Therefore, considering real-world constraints defined by the underlying QoS of the network presents an interesting problem of developing a practical and comprehensive stochastic model  so as to quantify secrecy performance in a realistic manner.  
\subsection{Key Contributions}
This work, aimed at filling the mentioned research gap, has the following key contributions:
\begin{itemize}
%{
 \item 	{We have proposed a novel generalised QoS-aware stochastic spatial PLS to investigate secrecy analysis in the presence of an attacker with both capabilities- eavesdropping and jamming.} Here, the key aspect is that our model allows any arbitrary location for nodes within the deployment region \textit{(Section 2)}.%
%We extend the proposed model  to investigate the impact of the attacker as a jammer on the secrecy performance of the given system 
\item First-time characterisation of attacker's eavesdropping capability is provided in terms of novel eavesdropping zone concept. The proposed model also incorporates the restriction being imposed  on the maximum distance between legitimate nodes arising out of the underlying QoS requirement. Considering these practical constraints, the distribution functions of the distances and SNRs of legitimate, eavesdropping and jamming links are derived. Closed-form expressions are also derived for ratio distributions of  legitimate-to-attacker link SNRs. This is a vital figure of merit being applied to derive the secrecy performance measures \textit{(Section 3)}.
\item Using probabilistic distance based distributions, %derived by relaxing the PPP assumptions, 
we obtain expression for secrecy outage probability 
having considered attacker's eavesdropping capability.  As such distribution functions could be used in a class of generalised practical scenarios.   We also derived closed-form expressions  for  SOP in asymptotic scenario \textit{(Section 4)}.
\item  Probabilistic characterisation of  secrecy outage  is also provided  having considered attacker's jamming capability. For this case also, SOP is derived in closed-form for asymptotic case \textit{(Section 4)}.
 \item A generalised framework is provided  for different possible deployment configurations for legitimate nodes and attacker. Also, novel and significant analytical insights  about designing of different system parameters  are provided 
under eavesdropping  and jamming attacks.
Maximum allowed separation between legitimate nodes are determined for achieving desired  secrecy performance  from the user's perspective, while the optimal value of eavesdropping range and jamming power are investigated from the attacker's perspective. \textit{(Section 5)}.
 \item Numerical results validate the proposed analysis  and present secrecy-aware key design perspective  by analysing the impact of inter-node distances on secrecy performance.  The relative severity of eavesdropping and jamming capabilities of attacker are also compared. It is shown that attacker must  have more transmit power  than legitimate source to  create more adverse conditions during jamming as compared to eavesdropping provided that it is capable to eavesdrop in entire  deployment range
% to pose more secrecy outage in jamming than eavesdropping provided that it is capable to eavesdrop in entire  deployment range.  
% A brief insight is also provided for relative severity of eavesdropping and jamming capabilities of attacker for the given system parameters 
\textit{(Section 6)}.
%\end{enumerate}
\end{itemize}
\subsection{Novelty and Scope}
 
%  adopts a more practical model by adopting 
%  probabilistic distance-based distribution   with consideration of the QoS-aware design parameters,instead of taking naive ppp assumptions,
% That considers  probabilistic distance-based d and 
To the best of our knowledge, this is the first work  \textit{ 
 that adopts a  probabilistic distance-based practical model with consideration of the real-time constraints in terms of QoS-controlling parameters  for secrecy analysis under eavesdropping as well as jamming. }
   We also propose \textit{eavesdropping zone concept} to incorporate the effect caused by eavesdropping capability of attacker
% %   We introduce \textit{eavesdropping zone concept} to include the effect of attacker`s eavesdropping capability 
 and  present 
   novel analysis on the impact of stochastic inter-node distances on secrecy outage.

  Scope of this work includes, though not limited to, development of a system with desired secrecy performance  by exploiting the randomness of inter-node distances. This work, providing insights on designing of parameters of legitimate nodes and attacker, can be extended to multiple-input multiple-output (MIMO) model.  Directional  beam-forming can  be applied to enhance the secrecy further. Additional insights can also be obtained by considering different fading environments. {Another important extension of the work may include the selection of appropriate attacking mode and optimisation of designing parameters of attacker  with the objective of minimising the achievable secrecy transmission rate.}
Though the widespread utility of the derived closed-form expressions for secrecy analysis is constrained by high SNR, they provide lower bounds on SOP  with the positive SC to give useful insights. There is a class of practical applications that can benefit from proposed designs. These include small cell networks, IoT networks, ad-hoc networks where  nodes are deployed within a small cell area. Hence, the impacts of leakage due to wiretapping and  interference due to jamming  are much greater than noise in the channel.
\section{ Proposed System Model}\label{sec:systm model}
   {In this section, we present the  proposed system model including the spatially stochastic network   topology, QoS-aware PLS model, channel model along with link SNR and SC  definitions for background setup of the problem for both eavesdropping and jamming  modes of attacker. }
 \subsection{Network topology}
{We consider a secure communication IoT scenario comprising of uniformly deployed multiple source-user pairs, where  each pair is allocated orthogonal resources in terms of a dedicated time slot, or set of sub-carriers. The ongoing communication is assumed to be under security threat from a hybrid  attacker with half-duplex capability. %with its exact location unknown.
It  can over-hear ongoing transmission as an eavesdropper through wiretapping or can act as a jammer causing interference to the information flow \cite{Wyner_Encoding,Garnaev, MyWCL}.
It  is further assumed to be present within deployment region with some statistical information about its location. However, 
%the exact distance between attacker and any source or user is unknown. 
the  relative distances between source, user and attacker are  unknown to each other exactly.}
%The exact distance between attacker and any source or user is assumed to be unknown while only statistical information is present about location of attacker such that exact distance b 
% Though attacker is assumed to be within deployment region but the exact distance between legitimate nodes and attacker is unknown
{Adopting this orthogonal multi-access IoT setup, we can focus on any randomly chosen source-user pair  and investigate the impact of attacker’s presence on the  ongoing secure transmission. It may be noted that this discussion and the proceeding analysis holds for any randomly-selected  source-user pair. Here, such an information source $\mathcal{S}$ and legitimate user $\mathcal{U}$  are assumed to be spatially distributed with uniform distribution in deployment region of radius $R$. Whereas, the attacker $\mathcal{A}$ can be randomly located anywhere in the deployment region. }{ Without loss of generality, we  explore secrecy performance 
    %conduct  our novel investigation,
    assuming $\mathcal{A}$ to be located at the origin, %first in eavesdropping mode, and then in jamming mode 
    %to explore secrecy performance 
    as depicted in Fig. \ref{fig: system_model}.}  {Novel investigations are  carried out under eavesdropping and jamming mode separately.
    % Here, $\mathcal{A}$ can also  be disguised as an IoT device for invading the network with malicious intentions.
    Here, $\mathcal{A}$ is also considered as an IoT device so as to invade the network  with malicious intentions by disguising itself. Considering form-factor constraints and lower hardware complexity of IoT devices, all nodes, including attacker, are assumed to be equipped with a single antenna \cite{HW_R23,HW_R24,HW_R25,HW_R26}.}
   
 \begin{figure*}[!t]    
      \begin{minipage}{.48\textwidth}    
        \centering
        \centering
       % {{\includegraphics[width=2.5in]{Figures/Revised_eve.eps} }}
        {{\includegraphics[width=1.7in]{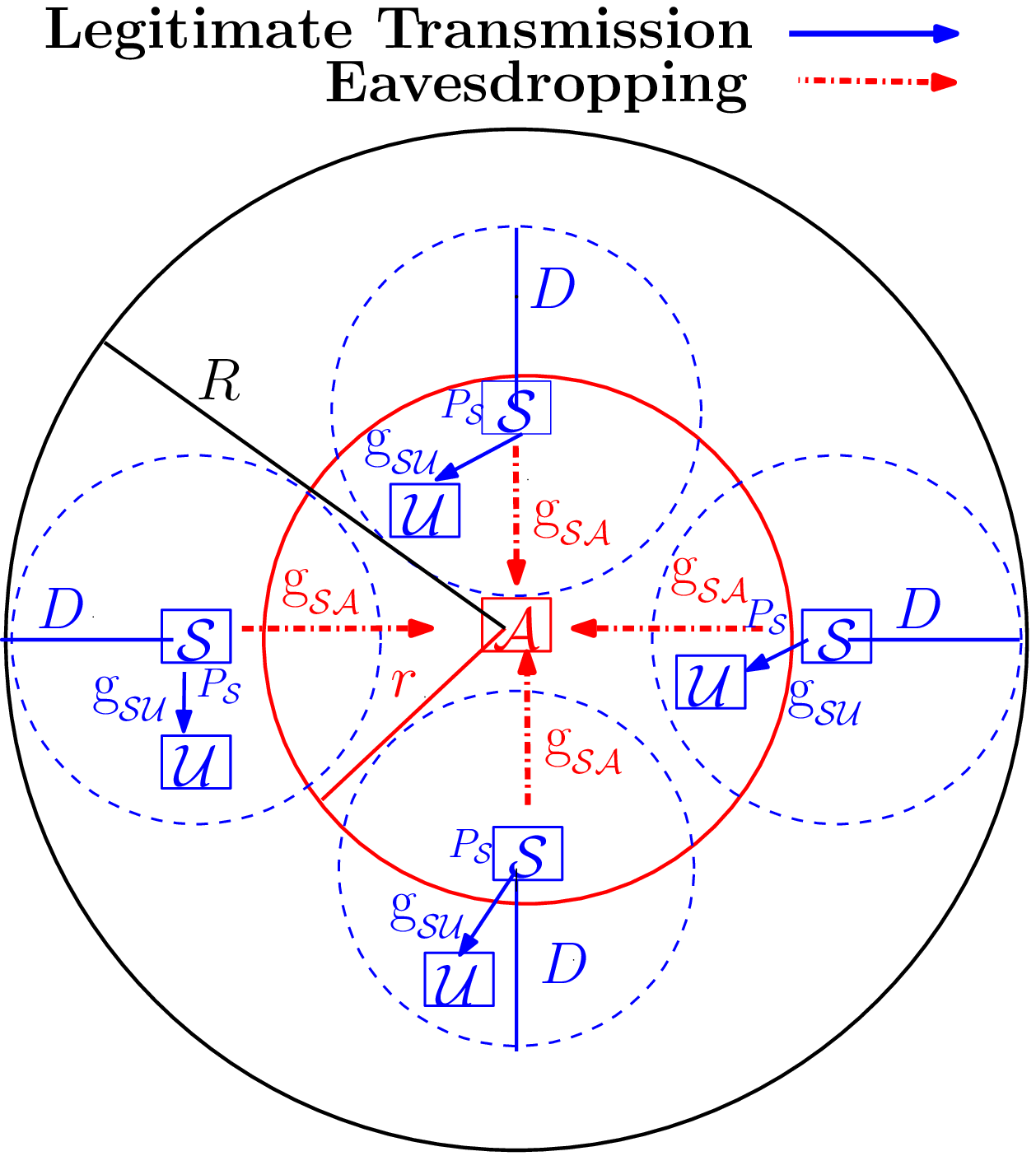} }}
         \end{minipage}
    %\quad
    \begin{minipage}{.48\textwidth}        
        \centering
          {{\includegraphics[width=1.7in]{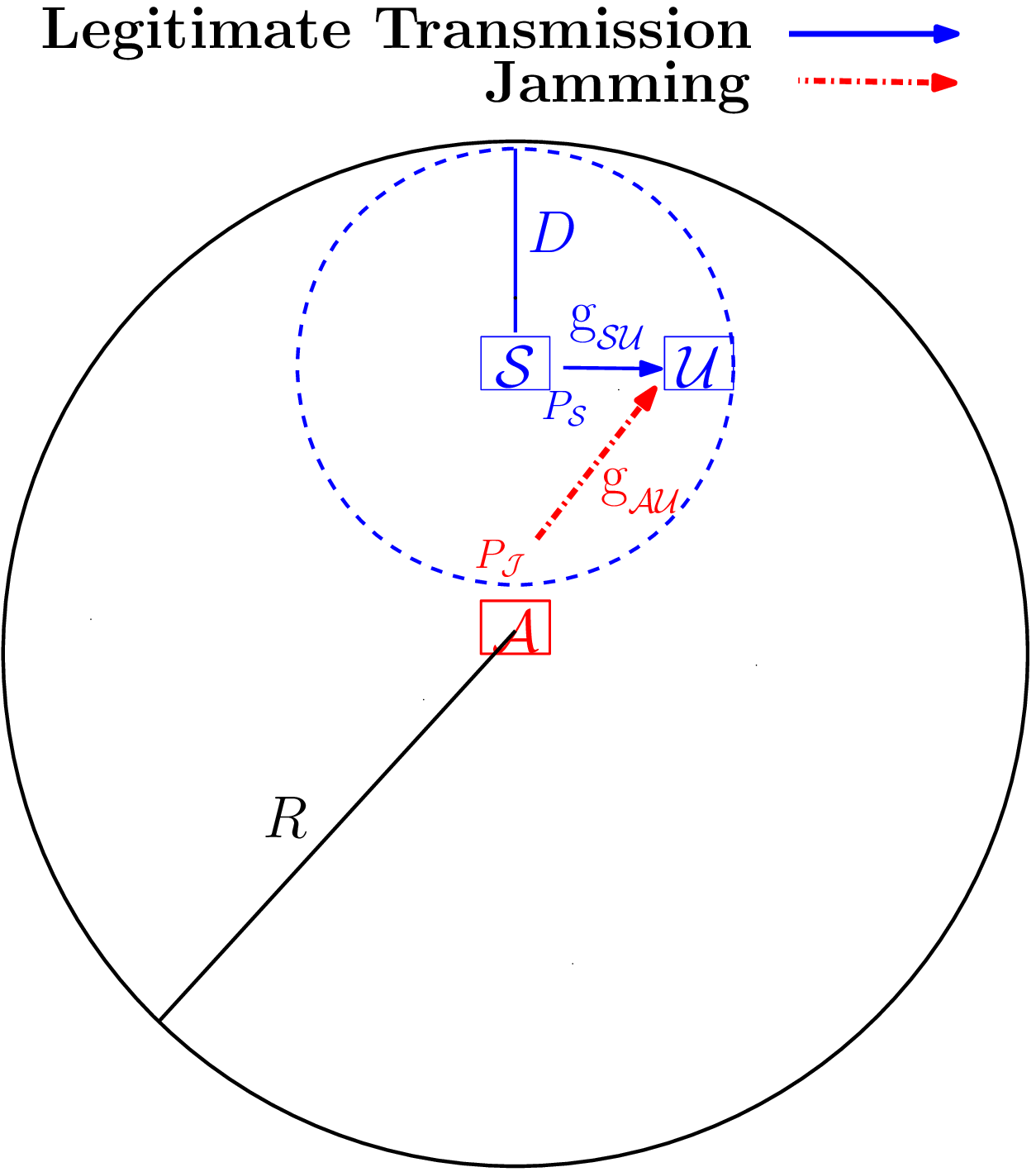} }}
       \end{minipage} 
     \caption{\footnotesize \textit{ QoS-aware stochastic spatial PLS model  with attacker in eavesdropping and jamming mode.} }
       \label{fig: system_model} \vspace{-4mm}
\end{figure*}  
\subsection{QoS-Aware PLS  Model}
%The proposed QoS-Aware PLS  Model investigates the PLS under eavesdroping as well as jamming attacks with consideration of the real-world constraints defined by the underlying QoS.
{The proposed  model investigates the  PLS  based secrecy metrics under eavesdropping as well as jamming. Here the phrase  ``QoS-aware" 	is associated with the obtaining of secrecy QoS requirements through
	critical design parameters of user as well as attacker from their respective point of view. These parameters include the distance threshold $D$, deployment range $R$ and source power 	$P_\mathcal{S}$ from the user's perspective. On the other hand, attacker's parameters like eavesdropping range $r $ and  jamming power $ P_{\mathcal{J}} $ also represent the capabilities to impact the network secrecy performance. Hence, they reflect the indirect measure of QoS attribute from attacker’s perspective.
	The proposed model first time characterises the secrecy performance by considering the real-world constraints in terms of these  parameters defined by the underlying QoS requirements. Following this, %With the aim of developing a practical model, 
 the maximum possible geographical distance between $\mathcal{S}$  and $\mathcal{U}$ has been restricted  to $ D$ determined  by the required QoS.  The distance parameter $D$ is utilised to ensure the minimum QoS requirement  over $\mathcal{S}$- $\mathcal{U}$ link. Further, for taking hardware limitations of $\mathcal{A}$ into account, we consider that under eavesdropping mode, the $\mathcal{A}$ remains effective in the disc area defined as eavesdropping zone  with centre  $\mathcal{A}$ and radius $r$. Here, $r$ is defined as the eavesdropping range of $\mathcal{A}$. It represents the maximum distance, to which $\mathcal{A}$ can eavesdrop the legitimate signal.} We have categorised all possible deployment scenario  of  $\mathcal{S}$ and $\mathcal{U}$  with respect to  $\mathcal{A}$'s eavesdropping zone   into the  four cases as illustrated  in Fig. \ref{fig: system_model}. These are:
 \begin{itemize}
  \item \textit{Case 1:} Both $\mathcal{S}$ and $\mathcal{U}$  are deployed within  the $\mathcal{A}$'s eavesdropping zone, i.e., ${d_{_{\mathcal{S}\mathcal{A}}}} <r $ and ${d_{_{\mathcal{U}\mathcal{A}}}} <r. $

    \item \textit{Case 2 :} $\mathcal{S}$  is deployed within   $\mathcal{A}$'s eavesdropping zone  while  $\mathcal{U}$ is  outside, i.e., ${d_{_{\mathcal{S}\mathcal{A}}}} <r $ and $ r \leq {d_{_{\mathcal{U}\mathcal{A}}}} <R. $
{
    \item \textit{Case 3 :}  $\mathcal{U}$  is deployed within   $\mathcal{A}$'s eavesdropping zone  while  $\mathcal{S}$ is outside,  i.e., $ r \leq {d_{_{\mathcal{S}\mathcal{A}}}} < R $ and ${d_{_{\mathcal{U}\mathcal{A}}}} <r. $
    \item \textit{Case 4 :} Both $\mathcal{S}$ and $\mathcal{U}$ both are deployed outside   $\mathcal{A}$'s eavesdropping zone, i.e., $ r\leq {d_{_{\mathcal{S}\mathcal{A}}}} <R $ and $r\leq {d_{_{\mathcal{U}\mathcal{A}}}} <R. $}\\
      \end{itemize} 
      %\vspace{-6mm}
      Here,  the distance between node $i$ and $j$ is represented by  $d_{_{ij}}, \;\forall\; i,j=\left\lbrace\mathcal{S},\mathcal{U},\mathcal{A}\right\rbrace$. 
      
          On the other hand, under jamming mode, $\mathcal{A}$ sends its signal with jamming  power  $ P_{\mathcal{J}} $ as interference to degrade the legitimate reception. 
{In analogy with eavesdropping mode, where $r$ controls the strength of attacker, $P_{\mathcal{J}}$ plays the same role and regulates the attacking capability in case of jamming. Hence, the effective region of the $\mathcal{A}$ depends upon $P_{\mathcal{J}}$ in jamming mode.}
%%%%%%%%%%%%%
% It is to mention here that $r$ and $ P_{\mathcal{J}} $ represent $\mathcal{A}$'s capability to impact the network performance. Hence, they reflect the indirect measure of QoS attribute from $\mathcal{A}$'s perspective. The deployment range  $R$ is also determined by network designers as per QoS requirement.
%%%%%%%%%%
%     \item \textit{Case 2 :} $\mathcal{S}$  lies inside   $\mathcal{A}$`s active region  while  $\mathcal{U}$ remains outside  i.e. ${d_{_{\mathcal{S}\mathcal{A}}}} <r $ and ${d_{_{\mathcal{U}\mathcal{A}}}} >r $

\subsection {Channel Model}
All the links and hence their SNR are  assumed to be independent due to different transmitting and receiving antenna gains, polarisation losses and small-scale fading at legitimate nodes and attacker \cite{Pinto08}.
%and subject to  path-loss  and quasi-static fading 
% The  channel gain  associated with  the link between node $i$ and node $j$ is represented by $\mathrm{g}_{ij}$.
The  channel gain  of  the link between node $i$ and  $j$ is represented by $\mathrm{g}_{ij}$ %in general 
and  modelled as:
% ad , as , and ar
% account for the channel parameters, namely, fading and antenna
% gains, in the respective link, and l is the path loss exponent
%   \begin{equation}\label{eq: gain}
%  \left|\mathrm{g}_{_{ij}}\right|^2= \frac{{\left|\mathrm{h}_{_{ij}}\right|^2}}{\left(d_{_{ij}}\right)^\theta} \qquad\forall i,j=\left\lbrace\mathcal{S},\mathcal{U},\mathcal{A}\right\rbrace.
% \end{equation}
 \begin{equation}\label{eq: gain}
 \left|\mathrm{g}_{_{ij}}\right|^2= \frac{a_{_{ij}}}{\left(d_{_{ij}}\right)^\theta}, \qquad\forall\; i,j=\left\lbrace\mathcal{S},\mathcal{U},\mathcal{A}\right\rbrace.
\end{equation}
Here $\theta$ is path-loss exponent,% $d_{_{ij}}$ corresponds to the distance between node $i$ and $j$, 
and $a_{_{ij}}$ accounts for the channel parameters like fading and antenna
 gains of the link between node $i$ and  $j$.
% Here $\theta$ is path-loss exponent; $d_{_{ij}}$ corresponds to the distance between node $i$ and $j$, ${\mathrm{h}_{_{ij}}}$ is channel gain coefficient 
%complex fading coefficient 
%of the link node $i$ and node $j$.% assumed constant (quasi-static) during the communications \cite{Pinto08}. 
{ It is to emphasise that %we have considered large-scale fading in the development of analytical framework because 
being an external entity, 
$\mathcal{A}$ does not cooperate with IoT setup towards revealing its location. It is also a well-known fact that when the exact distance between communicating nodes $d_{_{ij}}$ is  unknown, the large-scale
fading has a dominating impact whereas small-scale fading revolves around path-loss \cite[Fig. 2.1]{Goldsmith}. Therefore, for analytical tractability,
 we consider mean channel fading gain of the respective link to incorporate the small-scale  fading \cite{Pinto08, Pinto1} and focus on random long term fading in the development of analytical framework in detail leaving the joint  analysis of both as a separate investigation for future work. While  for comprehensive performance evaluation,  we analyse the impact of  small-scale fading randomness on designing of system parameter via simulation in result section.}
{For realistic analysis, we consider that  only statistical information is available about CSI for all the links.}
\subsection {Background Setup of the Problem}
For eavesdropping mode, the  received signals   ${y}_{_{\mathcal{S}\mathcal{U}}}$ and ${y}_{_{\mathcal{S}\mathcal{A}}}$ by $\mathcal{U}$ and  $ \mathcal{A}$ are represented by:
\begin{equation}\label{eq: signal_E}
{y}_{_{\mathcal{S}\mathcal{N}}}={x}_{\mathcal{S}}\sqrt{P_{\mathcal{S}}}\mathrm{g}_{_{\mathcal{S}\mathcal{N}}}+{w}{_{\mathcal{N}}},\qquad \mathcal{N}\in\left\lbrace\mathcal{U},\mathcal{A}\right\rbrace,
\end{equation}
where ${x}_{{\mathcal{S}}}$ is the zero mean and unit variance  signal transmitted by $\mathcal{S}$, and  $P_{{\mathcal{S}}}$ is the transmit power of $\mathcal{S}$. 
In \eqref{eq: signal_E}, $\mathrm{g}_{_{\mathcal{S}\mathcal{N}}}$ represents channel gain for legitimate  and eavesdropping links respectively for ${\mathcal{N}} \in \left\lbrace\mathcal{U},\mathcal{A}\right\rbrace $ where legitimate link refers for $\mathcal{S}$-$\mathcal{U}$ link and eavesdropping  link refers for $\mathcal{S}$-$\mathcal{A}$ link. Lastly, ${w}{_{\mathcal{N}}}$ represents zero mean additive white Gaussian noise with variance $\sigma^2$ as received at node ${{\mathcal{N}}}$. Without loss of generality, we simply assume that they are identical.
%{We now consider a special case where the randomness in channel gain is introduced by path-loss only, hence we consider average fading coefficient of the corresponding link to include the fading effects.}
 The corresponding SNR of legitimate link and eavesdropping link can be obtained as:  
\begin{align}\label{eq: snr_E}
\gamma_{_{\mathcal{S}\mathcal{N}}}=\frac{P_{\mathcal{S}}\vert \mathrm{g}_{_{\mathcal{S}\mathcal{N}}}\vert ^{2}}{\sigma^{2}}%=\frac{P_{\mathcal{S}} a_{_{\mathcal{S}\mathcal{N}}}}{\sigma^{2}{(d_{_{\mathcal{S}\mathcal{N}}}})^{\theta}}
= \frac{\kappa_{_{\mathcal{S}\mathcal{N}}}} {({d_{_{\mathcal{S}\mathcal{N}}}})^{\theta}},  %\mathcal{N}\in\left\lbrace\mathcal{A},\mathcal{U}\right\rbrace, 
\;\text{with}\;  \kappa_{_{\mathcal{S}\mathcal{N}}}\triangleq \frac{P_{\mathcal{S}}a_{_{\mathcal{S}\mathcal{N}}}}{\sigma^{2}}.
\end{align}
%where $ \kappa_{_{\mathcal{S}\mathcal{N}}}\triangleq \frac{P_{\mathcal{S}}a_{_{\mathcal{S}\mathcal{N}}}}{\sigma^{2}}$.
%and $a_{_{\mathcal{S}\mathcal{N}}}$ is average channel fading gain of the corresponding link.
 The maximum achievable rate for the transmission in the presence of  the $\mathcal{A}$  as an eavesdropper is  given by the SC and defined for eavesdropping attack  as \cite {Wyner}:
 \begin{align}\label{eq: C_E}
 \hspace{-3mm}C_{s}^{{{\mathcal{E}}}}  \hspace{-2mm}=\left[\log_{2}(1+ \gamma_{_{\mathcal{S}\mathcal{U}}})- \log_{2}(1+ \gamma_{_{\mathcal{S}\mathcal{A}}})\right]^{+}   \hspace{-1mm}=\left[\log_{2}\frac{1+ \gamma_{_{\mathcal{S}\mathcal{U}}}}{1+ \gamma_{_{\mathcal{S}\mathcal{A}}}}\right]^{+}\hspace{-2mm},%\hspace{-2mm}
 \end{align}
 {where} $[x]^+ = \max [{x, 0}].$
 %where% $C_{_{\mathcal{S}\mathcal{U}}}$ and $C_{_{\mathcal{S}\mathcal{A}})$ are, respectively, the legitimate transmission and wiretap rate, and
Alternatively, when  the $\mathcal{A}$ works as a jammer to disrupt the legitimate channel, the signal ${y}_{_{\mathcal{S}\mathcal{U}}}$ received by $ \mathcal{U}$  can be represented as:
 \begin{equation}\label{eq: signal_J}
{y}_{_{\mathcal{S}\mathcal{U}}}={x}_{\mathcal{S}}\sqrt{P_{\mathcal{S}}}\mathrm{g}_{_{\mathcal{S}\mathcal{U}}}+{x}_{\mathcal{J}}\sqrt{P_{\mathcal{J}}}\mathrm{g}_{_{\mathcal{A}\mathcal{U}}}+{w}{_{\mathcal{U}}},
\end{equation}
 where ${x}_{{\mathcal{J}}}$ is the zero mean and unit variance  signal transmitted by $\mathcal{A}$, and  $P_{{\mathcal{J}}}$ is the jamming power of $\mathcal{A}$. The channel gain associated with jamming link is represented as:   $\mathrm{g}_{_{\mathcal{A}\mathcal{U}}}$  as defined above. Here, jamming link refers to  the $\mathcal{A }- \mathcal{U}$ link. When $\mathcal{A}$ is in jamming mode, the maximum achievable rate  is defined as SC and given as \cite{Poor11, Ryu}:  
    \begin{align}\label{eq: C_J} \hspace{-3mm}C_{s}^{{{\mathcal{J}}}}  \hspace{-1mm}=\log_{2}\left(1+ \frac{P_{\mathcal{S}}  \left|\mathrm{g}_{_{\mathcal{S}\mathcal{U}}}\right|^2}{P_{J} \left|\mathrm{g}_{_{\mathcal{A}\mathcal{D}}}\right|^2 +\sigma_{0}^{2}}\right) \hspace{-1mm}=\log_{2}\left(1+ \frac{\gamma_{_{\mathcal{S}\mathcal{U}}}}{1+ \gamma_{_{\mathcal{A}\mathcal{U}}}}\right)\hspace{-1mm}. \hspace{-2mm}\end{align}
     where $\gamma_{_{\mathcal{A}\mathcal{U}}}$ is SNR of jamming link. It can be obtained as:  
\begin{align}\label{eq: snr_J}
\gamma_{_{\mathcal{A}\mathcal{U}}}=\frac{P_{\mathcal{J}}\vert \mathrm{g}_{_{\mathcal{A}\mathcal{U}}}\vert ^{2}}{\sigma^{2}}%=\frac{P_{\mathcal{J}}\overline{\vert\mathrm{h}_{_{\mathcal{A}\mathcal{U}}}\vert ^{2}}}{\sigma^{2}{(d_{_{\mathcal{A}\mathcal{U}}}})^{\theta}}
= \frac{\kappa_{_{\mathcal{A}\mathcal{U}}}} {({d_{_{\mathcal{A}\mathcal{U}}}})^{\theta}},\; \text{with}\; \kappa_{_{\mathcal{A}\mathcal{U}}}= \frac{P_{\mathcal{J}}a_{_{\mathcal{A}\mathcal{U}}}}{\sigma^{2}}.
\end{align} 

%and $\overline{\vert h_{_{\mathcal{A}\mathcal{U}}}\vert ^{2}}$ is average fading coefficient of $\mathcal{A }- \mathcal{U}$ link. 
\section{ Distribution Functions of Stochastic Distances and SNRs }
 \label{sec:RatioDistribution} 
{ This section presents the   novel QoS-aware distance  and SNR distributions for legitimate link $\mathcal{S}$-$\mathcal{U}$,   eavesdropping link $\mathcal{S}$-$\mathcal{A}$, and jamming link $\mathcal{A}$-$\mathcal{U}$ using the disk point picking and disk line picking  distributions given by a new  geometric probability technique \cite{Fischbach2000}. These distributions enable us to derive corresponding distance  and SNR distributions for any deployment configuration in our generalised model. 
 %Since our model offers adaptability for a node having  any random location, we may derive distance  and SNR distributions for legitimate link $\mathcal{S}$-$\mathcal{U}$,   eavesdropping link $\mathcal{S}$-$\mathcal{A}$ and jamming link $\mathcal{A}$-$\mathcal{U}$  in any deployment configuration. 
 For instance,  if $\mathcal{S}$ is considered at origin, and $\mathcal{U}$ and $\mathcal{A}$ are randomly deployed, then the distances of  $\mathcal{S}$-$\mathcal{U}$ and  $\mathcal{S}$-$\mathcal{A}$ links  follow disk point picking distribution while   distance of $\mathcal{A}$-$\mathcal{U}$ link follows disk line picking distribution.  Alternatively, if all the three nodes are randomly deployed, distances of all the links  follow the  disk line picking distributions. This discussion is more elaborated in Section 5. Thus, here we derive the corresponding distributions for configuration considered in the proposed system model.
%In this section, we derive  distance  and SNR distributions for legitimate link $\mathcal{S}$-$\mathcal{U}$,   eavesdropping link $\mathcal{S}$-$\mathcal{A}$ and jamming link $\mathcal{A}$-$\mathcal{U}$.
Additionally, ratio distributions of $\mathcal{S}$-$\mathcal{U}$ -to- $\mathcal{S}$-$\mathcal{A}$ link SNRs and  $\mathcal{S}$-$\mathcal{U}$ -to- $\mathcal{A}$-$\mathcal{U}$ link SNRs are also derived with no loss of generality. }
\subsection {Distance Distributions  for Legitimate, Eavesdropping and Jamming Links }
 To obtain the distributions of SNR for the legitimate, eavesdropping and jamming links, we first investigate the corresponding distance distributions defined as  PDF $ f_{d_{_{\mathcal{S}\mathcal{U}}}}(l)$ of distance  ${d_{_{\mathcal{S}\mathcal{U}}}}$, PDF $ f_{d_{_{\mathcal{S}\mathcal{A}}}}(l)$ of distance  ${d_{_{\mathcal{S}\mathcal{A}}}}$ and PDF $ f_{d_{_{\mathcal{A}\mathcal{U}}}}(l)$ of distance  ${d_{_{\mathcal{A}\mathcal{U}}}}$ as follows: 
  \subsubsection{$f_{d_{_{\mathcal{S}\mathcal{U}}}}(l)$} It is  provided by Proposition 1 with the consideration of practical constraint  on maximum distance between $\mathcal{S}$  and $\mathcal{U}$.
    \begin{proposition}  The PDF $ f_{d_{_{\mathcal{S}\mathcal{U}}}}(l)$ of distance  ${d_{_{\mathcal{S}\mathcal{U}}}}$ is given below subject to the condition that the maximum distance between $\mathcal{S}$  and $\mathcal{U}$  is restricted to $D$.
   \begin{align}\label{eq:PDFSU}
   & f_{d_{_{\mathcal{S}\mathcal{U}}}}(l|l\leq D )\nonumber\\\triangleq 
    &\begin{cases}\frac{2l}{F_{d_{\mathcal{S}\mathcal{U}}}(D)R^2} \left({1- \frac{B_{\frac{l^2}{4 R^2}}\left(\frac{1}{2},\frac{3}{2}\right)}{{B\left(\frac{3}{2},\frac{1}{2}\right)}}}\right), &
   % \text{$d_{_{\mathcal{S}\mathcal{U}}}<D\leq 2R,  $}\\
     \text{$ l<D\leq 2R,  $}\\
0,   & \text{otherwise.}
    \end{cases}
    \end{align} 
     where $ F_{d_{_{\mathcal{S}\mathcal{U}}}(D)}= \text{Pr}(l<D)=\frac{D^2}{R^2}-\frac{D^2}{{R^2B\left(\frac{3}{2},\frac{1}{2}\right)}}B_{\frac{D^2}{R^2}}\left(\frac{3}{2},\frac{1}{2}\right)-B_{\frac{D^2}{R^2}}\left(\frac{3}{2},\frac{3}{2}\right)$, $\mathrm{B}_x(p,q)  =\int_0^x t^{p-1} (1-t)^{q-1} \mathrm{d}t $ is an incomplete beta function  and $\mathrm{B}(p,q)= \int_0^1 t^{p-1} (1-t)^{q-1} \mathrm{d}t $  is an complete beta function.
\end{proposition}
 \begin{proof}
 {Given an $n$-dimensional ball of radius $R$, disk line picking distribution i.e., the distribution of the distances between  two points chosen at random within the ball is given by a  new  geometric  probability  technique  \cite[eq.(28)]{Fischbach2000}.  For circle of radius $R$, i.e., a special case of $n=2$,  PDF for the distance between  two   random points representing nodes $\mathcal{S}$ and $\mathcal{U}$  with uniform node distribution,    is reduced to  
       \begin{align}\label{eq:PDFSU1}
    f_{d_{_{\mathcal{S}\mathcal{U}}}}(l)= \frac{2l I_{1-\frac{l^2}{4 R^2}}\left(\frac{3}{2},\frac{1}{2}\right) }{ R^2}, \;\;\;\;\;
   % \text{$d_{_{\mathcal{S}\mathcal{U}}}<D\leq 2R,  $}\\
     \text{$ 0<l\leq 2R,$}
    \end{align} 
      where $  I_x(a,b)=\frac{B_x(a,b) }{B(a,b)}$  is defined as a regularised beta function \cite[eq.(29)]{Fischbach2000}. }
In this work, the maximum distance between $\mathcal{S}$ and $\mathcal{U}$  is limited to $D$. Thus, the required distribution turns into a right truncated distribution  which can be obtained by limiting the domain of ${d_{_{\mathcal{S}\mathcal{U}}}} $ to $D$  
  and re-normalising the $f_{d_{_{\mathcal{S}\mathcal{U}}}} $ to satisfy $\int^D_0 f_{d_{\mathcal{S}\mathcal{U}}}(l,R | l\leq D ) dl=1 $. Hence, the desired truncated function results in  $f_{d_{\mathcal{S}\mathcal{U}}}(l|l\leq D )\triangleq \frac{f_{d_{\mathcal{S}\mathcal{U}}}(l)}{F_{d_{\mathcal{S}\mathcal{U}}}(D)}$ where   $ F_{d_{\mathcal{S}\mathcal{U}}}(D) =\text{Pr}(l<D)$, is normalising  factor. Lastly, by using   identity of  beta function  $\mathrm{B}_z(a,b)=\mathrm{B}(a,b)-\mathrm{B}_{1-z}(b,a) $ \cite[eq. (06.19.17.0008.01)]{betaintegral},
  we find   $f_{d_{_{\mathcal{S}\mathcal{U}}}}$ as given by \eqref{eq:PDFSU}.
    \end{proof}
 \subsubsection{$f_{d_{_{\mathcal{S}\mathcal{A}}}}(l)$ and  $f_{d_{_{\mathcal{A}\mathcal{U}}}}(l)$}

{Given a circle having radius $R$ with uniform node distribution,  the PDF for the distance of a point from the
centre is well known in literature \cite[eq.(20)]{Omiyi}  as disk point picking distribution. Following that, The PDF $ f_{d_{_{\mathcal{S}\mathcal{A}}}}(l)$ of distance  ${d_{_{\mathcal{S}\mathcal{A}}}}$ and $ f_{d_{_{\mathcal{A}\mathcal{U}}}}(l)$ of distance  ${d_{_{\mathcal{A}\mathcal{U}}}}$   is given by:
  \begin{align}\label{eq:PDFAU}
{ f_{d_{_{\mathcal{N}\mathcal{A}}}}(l) \triangleq   \frac{2 l}{R^2}},  \quad\quad \  & \text{$ 0<l\leq  R \ $},
\end{align}
where $\mathcal{N}  \in \left\lbrace\mathcal{S},\mathcal{U}\right\rbrace$.
It is to note that  ${d_{_{\mathcal{N}\mathcal{A}}}}$ and ${d_{_{\mathcal{A}\mathcal{N}}}}$ are the same in this work. }
\subsection {SNR Distributions  for Legitimate, Eavesdropping and Jamming Links }
To obtain the  ratio distribution of SNR of legitimate link and attacker link, we first investigate the  SNR distribution for individual link including  PDF $ f_{\gamma_{_{\mathcal{S}\mathcal{U}}}}(x)$ for legitimate link SNR, PDF $ f_{\gamma_{_{\mathcal{S}\mathcal{A}}}}(y)$ for eavesdropping link SNR and PDF $ f_{\gamma_{_{\mathcal{A}\mathcal{U}}}}(y)$ for jamming link SNR.
\subsubsection{$f_{\gamma_{_{\mathcal{S}\mathcal{U}}}}\left(x\right)$}
As observed from \eqref{eq: snr_E}, $\gamma_{_{\mathcal{S}\mathcal{U}}}$ 
is a function of  random variable ${d_{_{\mathcal{S}\mathcal{U}}}}$, its PDF is derived by 
applying random variable transformation on \eqref{eq:PDFSU}   provided  ${d_{_{\mathcal{S}\mathcal{U}}}}$ is less than $D$ as follows:
 \begin{align}\label{eq:PDFSNRSU}
\text{$f_{\gamma_{_{\mathcal{S}\mathcal{U}}}}\left(x\right)\triangleq \begin{cases}
 \frac{8u \left(\pi-2 \mathrm{B}_{u}\left(\frac{1}{2},\frac{3}{2}\right)\right)}{ F_{d_{_{\mathcal{S}\mathcal{U}}}}(D) x \pi \theta   },\; & \text{$x>\frac{\kappa_{_{_{\mathcal{S}\mathcal{U}}}}}{{D}^{\theta}}, $}\\
 0,  & \text{otherwise, }
\end{cases}$}
\end{align} 
where $u= {\frac {\kappa_{_{_{\mathcal{S}\mathcal{U}}}}^{\frac{2}{\theta}}}{4R^2{x}^{\frac{2}{\theta}}}}$.

\subsubsection{$f_{\gamma_{_{\mathcal{S}\mathcal{A}}}}\left(y\right)$}
Under  eavesdropping mode, $\mathcal{A}$ tends to wiretap the signal transmitted by $\mathcal{S}$. Therefore, Proposition 2 presents the  distribution of SNR for  ${\mathcal{S}}- {\mathcal{A}}$ link with the consideration of restriction on eavesdropping capability.
\begin{proposition} When $\mathcal{A}$ is in eavesdropping mode, The PDF $ f_{\gamma_{_{\mathcal{S}\mathcal{A}}}}(x)$ of SNR  ${\gamma_{_{\mathcal{S}\mathcal{A}}}}$  is given below.
 \begin{align}\label{eq:PDFSNRSA}  
\text{{$f_{\gamma_{_{\mathcal{S}\mathcal{A}}}}\left(y\right) \triangleq \begin{cases}
\frac{2\alpha y^{-\frac{2}{\theta}-1} }{{\theta \kappa_{_{_{\mathcal{S}\mathcal{A}}}}^{-\frac{2}{\theta}} r^2}}, &\ \text{$y>\frac{\kappa_{_{_{\mathcal{S}\mathcal{A}}}}}{{r}^{\theta}},\ $}\\
(1-\alpha)\delta(y), &\ \text{$ y \leq \frac{\kappa_{_{_{\mathcal{S}\mathcal{A}}}}}{{r}^{\theta}},\ $} \\
\end{cases}$}}
\end{align}
 where $\alpha \triangleq r^2/R^2$ represents the probability that  $ \mathcal{S}$ lies within the eavesdropping zone  and $\delta(y)$ is used to denote the Dirac delta function.
 \end{proposition}
  \begin{proof} Generalised proposed model enables the legitimate nodes to be positioned  at any place in the circular deployment region having radius $R$. But  $\mathcal{A}$ has capability to  eavesdrop the legitimate transmission only when  $ \mathcal{S}$ lies within the eavesdropping zone  subject to ${d_{_{\mathcal{S}\mathcal{A}}}} <r $ (  illustrated  under cases 1 and 2). The probability of $ \mathcal{S}$ lying  inside the
  eavesdropping zone is denoted by $\alpha.$
 % This event occurs with a probability $\alpha.$
As observed from \eqref{eq: snr_E}, being   a function of  random variable ${d_{_{\mathcal{S}\mathcal{A}}}}$,    PDF of $\gamma_{_{\mathcal{S}\mathcal{A}}}$ can, therefore, be  derived with the help of  random variable transformation and \eqref{eq:PDFAU} satisfying $d_{_{_{\mathcal{S}\mathcal{A}}}}<r$  as follows:
  \begin{align}\label{eq:PDFSNRSAstd}
 f_{\gamma_{_{\mathcal{S}\mathcal{A}}}}(y) \triangleq \frac{2 \alpha}{{\theta y r^2}} \left(\frac{y}{\kappa_{_{_{\mathcal{S}\mathcal{A}}}}}\right)^{-\frac{2}{\theta}},   & \quad \quad \quad \text{$ y>\frac{\kappa_{_{_{\mathcal{S}\mathcal{A}}}}}{{r}^{\theta}}. \ $}
\end{align}
 In contrast, under the cases 3 and 4, %when $ \mathcal{S}$ lies outside the eavesdropping zone
 satisfying ${d_{_{\mathcal{S}\mathcal{A}}}} > r $, $\mathcal{A}$ is unable for eavesdropping the legitimate signal. For the purpose of realistic analysis, we may consider this scenario equivalent to the one when ${\gamma_{_{\mathcal{S}\mathcal{A}}}}$ approaches zero. This means that $\mathcal{A}$ has become ineffective.
%  For realistic analysis, this scenario  can, therefore, be considered equivalent to the one where ${\gamma_{_{\mathcal{S}\mathcal{A}}}}$ approaches zero i.e., $\mathcal{A}$ will become ineffective.
   Intuitively, this equivalence  is also supported by a matter of  fact that when  $ \mathcal{S}$  lies beyond the eavesdropping zone,  $\mathcal{A}$ will be unable to tap the signal transmitted by $\mathcal{S}$  regardless of the distance and SNR between them. We can express this equivalence mathematically as follows:
\begin{eqnarray}\label{eq:PDFSAstdelta}
f_{\gamma_{_{\mathcal{S}\mathcal{A}}}}(y) \triangleq  (1-\alpha)  \delta(y),   & \quad\quad\quad \text{$ y \leq \frac{\kappa_{_{_{\mathcal{S}\mathcal{A}}}}}{{r}^{\theta}}.  \ $}
 \end{eqnarray}
 By combining \eqref{eq:PDFSNRSAstd} and  \eqref{eq:PDFSAstdelta}, PDF of SNR of $\mathcal{S}$-$\mathcal{A}$ link can be obtained  as given in \eqref{eq:PDFSNRSA}.
 \end{proof}
 \subsubsection{$f_{\gamma_{_{\mathcal{A}\mathcal{U}}}}\left(y\right)$}
Under jamming mode, $\mathcal{A}$ transmits a signal with power $ P_{\mathcal{J}}$ to interfere with legitimate reception at $\mathcal{U}$. Therefore, we next present distribution of   SNR for  ${\mathcal{A}}-{\mathcal{U}}$ link.
$P_{\mathcal{J}} $ reflects the attacker capability of $\mathcal{A}$ as a jammer. Since $\mathcal{A}$ is not getting any feedback from $\mathcal{U}$,  it is not able to find that $\mathcal{U}$ is in its jamming range or not. %Hence, in case of jamming,     $P_{\mathcal{J}}$ is more appropriate parameter to co
     %aware about sensitivity of $\mathcal{U}$, hence it will not be able to calculate jamming range of attacker. 
Therefore, we are assuming that $\mathcal{A}$ can affect $\mathcal{U}$  located anywhere in the  region of deployment having radius $R$ based on its jamming power $ P_{\mathcal{J}}$.
 The PDF $ f_{\gamma_{_{\mathcal{A}\mathcal{U}}}}$(y) of SNR  ${\gamma_{_{\mathcal{A}\mathcal{U}}}}$ in jamming  attack is given by:
 \begin{align}\label{eq:PDFSNRAU}
\text{{$f_{\gamma_{_{\mathcal{A}\mathcal{U}}}}\left(y\right) \triangleq \begin{cases}
\frac{2 }{{\theta y R^2}}\left(\frac{y}{\kappa_{_{_{\mathcal{A}\mathcal{U}}}}}\right)^{-\frac{2}{\theta}}, & \text{$y>\frac{\kappa_{_{_{\mathcal{A}\mathcal{U}}}}}{{R}^{\theta}},\ $}\\
0,   & \text{otherwise.}
\end{cases}$}}
\end{align} 
\subsection {Ratio  distribution of  Legitimate to Attacker Link SNRs }
In this section, the  closed-form expressions for PDF ${{f}}_{_{\bar{\gamma}^{^ {\mathcal{S}\mathcal{U}}}_{_{\mathcal{S}\mathcal{A}}}}}\left(z\right)$ of the ratio  ${\bar{\gamma}^{^ {\mathcal{S}\mathcal{U}}}_{_{\mathcal{S}\mathcal{A}}}}\triangleq \frac{\gamma_{_{\mathcal{S}\mathcal{U}}}}{\gamma_{_{\mathcal{S}\mathcal{A}}}}$  for legitimate-to-eavesdropping link SNRs    and  PDF ${f}_{_{\bar{\gamma}^{^ {\mathcal{S}\mathcal{U}}}_{_{\mathcal{A}\mathcal{U}}}}}\left(z\right)$ of the ratio  ${\bar{\gamma}^{^ {\mathcal{S}\mathcal{U}}}_{_{\mathcal{A}\mathcal{U}}}} \triangleq \frac{\gamma_{_{\mathcal{S}\mathcal{U}}}}{\gamma_{_{\mathcal{A}\mathcal{U}}}}$ of legitimate-to-jamming link SNRs  and corresponding logarithmic transformations are provided. It is to be noted that $\gamma_{_{\mathcal{S}\mathcal{U}}}$ in \eqref{eq:PDFSNRSU} is derived  using underlying distance distribution given by a  distinct geometric probability technique \cite{Fischbach2000}. It  facilitates to obtain the proposed solutions in the closed-form, which was  challenging otherwise. These PDFs can be utilised as an important tool to obtain analytical tractable expressions for SOP through logarithmic transformations under eavesdropping as well as jamming   as described in next section. Therefore, the corresponding logarithmic transformation of ratio  PDFs are also provided as $f_{\log_2\left({\bar{\gamma}^{^ {\mathcal{S}\mathcal{U}}}_{_{\mathcal{S}\mathcal{A}}}}\right)} \left(c\right)$ in eavesdropping and $f_{\log_2\left(1+{\bar{\gamma}^{^ {\mathcal{S}\mathcal{U}}}_{_{\mathcal{A}\mathcal{U}}}}\right)} \left(c\right)$ in jamming.
\subsubsection{${{f}}_{_{\bar{\gamma}^{^ {\mathcal{S}\mathcal{U}}}_{_{\mathcal{S}\mathcal{A}}}}}\left(z\right)$}
Under  eavesdropping mode of $\mathcal{A}$, the cases 3 and 4 may exist with a probability of $(1-\alpha)$ 
 where  $\mathcal{S}$ may appear outside the  eavesdropping zone of $\mathcal{A}$. In this scenario,  ${\gamma_{_{\mathcal{S}\mathcal{A}}}}$ goes beyond an acceptable threshold to eavesdrop the signal;  $\mathcal{A}$'s channel, therefore, ceases to exist.
%$\mathcal{S}$ may reside outside  the  $\mathcal{A}$`s eavesdropping zone with  probability $(1-\alpha)$ Hence, in this case attacker`s channel ceases to exist.
This is mathematically represented by the \eqref{eq:PDFSNRSA}. Consequently, ratio  $\frac{\gamma_{_{\mathcal{S}\mathcal{U}}}}{{\gamma_{_{\mathcal{S}\mathcal{A}}}}}$  does not exist for the cases 3 and 4. Next, we introduce Lemma 1 to provide the PDF of ratio   ${\bar{\gamma}^{^ {\mathcal{S}\mathcal{U}}}_{_{\mathcal{S}\mathcal{A}}}}$ for the  cases 1 and 2.
\begin{lemma}\label{lemma: lemma1}
 The  PDF ${{f}}_{_{\bar{\gamma}^{^ {\mathcal{S}\mathcal{U}}}_{_{\mathcal{S}\mathcal{A}}}}}\left(z\right)$  of ratio ${\bar{\gamma}^{^ {\mathcal{S}\mathcal{U}}}_{_{\mathcal{S}\mathcal{A}}}}$ of  legitimate-to-eavesdropping link SNRs  under the conditions that $\mathcal{S}$ lies within the  eavesdropping zone of $\mathcal{A}$  is given by:
  \begin{align}\label{eq:PDFratio}
\begin{split}
{{f}}_{_{\bar{\gamma}^{^ {\mathcal{S}\mathcal{U}}}_{_{\mathcal{S}\mathcal{A}}}}}\left(z\right)\triangleq \begin{cases}
  \frac{ ({\lambda_{\mathcal{E}}})^{-\frac{2}{\theta}}D^4z^{\frac{2}{\theta}-1} }{ F_{d_{_{\mathcal{S}\mathcal{U}}}}(D) \theta R^2r^2} -\frac{ 32({{\lambda_{\mathcal{E}}}})^{-\frac{2}{\theta}}R^2 z^{\frac{2}{\theta}-1}\mathrm{\beta}_1(\frac{D^2}{4R^2})}{ F_{d_{_{\mathcal{S}\mathcal{U}}}}(D)  \pi \theta r^2},  \\\qquad \qquad \qquad \qquad \forall\;\text{ $ 0<z<\frac{{\lambda_{\mathcal{E}}}r^\theta}{{D}^{\theta}},\ $}\\
  \frac {({\lambda_{\mathcal{E}}})^{\frac{2}{\theta}} r^2z^{-\frac{2}{\theta}-1}}{  F_{d_{_{\mathcal{S}\mathcal{U}}}}(D)  \theta R^2} -\frac{32({\lambda_{\mathcal{E}}})^{-\frac{2}{\theta}}R^2 z^{\frac{2}{\theta}-1}\mathrm{\beta}_1 \left({\frac{({\lambda_{\mathcal{E}}})^{\frac{2}{\theta}} r^2}{z^{\frac{2}{\theta}}4R^2}}\right) }{ F_{d_{_{\mathcal{S}\mathcal{U}}}}(D)  \pi \theta r^2},\\ \qquad \qquad \qquad \qquad \forall\;  \text{  $   \frac{{\lambda_{\mathcal{E}}}r^\theta}{{D}^{\theta}}\ \leq  z <\infty   $},\\
  %%(1-\alpha)\delta(z-\infty), & \text{$ z=\infty $} \\
\end{cases}
\end{split}
\end{align}
where $ {\lambda_{\mathcal{E}}}=\left(\frac{\kappa_{_{_{\mathcal{S}\mathcal{U}}}}}{\kappa_{_{_{\mathcal{S}\mathcal{A}}}}}\right)$ and
$ \mathrm{\beta}_1(a)= a^2 \mathrm{B}_{a}(\frac{1}{2},\frac{3}{2})-\mathrm{B}_{a}\left(\frac{5}{2},\frac{3}{2}\right)$. \\
\end{lemma}
 \begin{proof}
 Using \cite[(eq. 6.60)]{populis} PDF of ${\bar{\gamma}^{^ {\mathcal{S}\mathcal{U}}}_{_{\mathcal{A}\mathcal{U}}}}$  for the cases 1 and 2 can be expressed as:
\begin{align}\label{eq: PDF_ratio2}
 &{f}_{\bar{\gamma}^{^ {\mathcal{S}\mathcal{U}}}_{_{\mathcal{S}\mathcal{A}}}}\left(z\right)
%\triangleq  \int\limits_{-\infty}^{\infty}\left|y\right| f_{\gamma_{_{\mathcal{S}\mathcal{U}}},{\gamma_{_{\mathcal{S}\mathcal{A}}}}}(zy,y)\mathrm{d}y\nonumber \\ 
 = \int\limits_{y_{l}}^{\infty}yf_{\gamma_{_{\mathcal{S}\mathcal{U}}}}(zy)f_{\gamma_{_{\mathcal{S}\mathcal{A}}}}\left(y|y>\frac{\kappa_{_{_{\mathcal{S}\mathcal{A}}}}}{{r}^{\theta}}\right)\mathrm{d}y\nonumber\\ &\stackrel{(a)}{=} \int\limits_{y_l}^{\infty}\frac{ 16 {\kappa_{_{_{\mathcal{S}\mathcal{A}}}}}^{\frac{2}{\theta}}y^{-\frac{2}{\theta}-1} u_0}{ F_{d_{_{\mathcal{S}\mathcal{U}}}}(D)\pi z r^2 \theta^2 }\left(\pi-2 B_{u_0}\left(\frac{1}{2},\frac{3}{2}\right)\right)\mathrm{d}y,%\hspace{-5mm}
\end{align} 
\text{where  $u_0\triangleq\frac {\kappa_{_{_{\mathcal{S}\mathcal{U}}}}^{\frac{2}{\theta}}}{4R^2{(zy)}^{\frac{2}{\theta}}}$}.
Here, $(a)$ is obtained from  definition of $f_{\gamma_{_{\mathcal{S}\mathcal{U}}}}\left(x\right)$ and  $f_{\gamma_{_{\mathcal{S}\mathcal{A}}}}\left(y\right )$ given in \eqref{eq:PDFSNRSU} and  \eqref{eq:PDFSNRSA}  for the cases 1 and 2 respectively.  Lower limit $ y_{_l}$ of ratio integral%  is  maximum of lower support of ${\gamma_{_{\mathcal{S}\mathcal{U}}}}$ and $ {\gamma_{_{\mathcal{S}\mathcal{A}}}}$ under the cases 1 and 2. It
is given by $\max\left (\frac{\kappa_{_{_{\mathcal{S}\mathcal{U}}}}}{{D}^\theta z},\frac{\kappa_{_{_{\mathcal{S}\mathcal{A}}}}}{r^\theta}\right)$ and  represented in   piece-wise  form as: 
\begin{align}
{y}_l=\begin{cases}
\frac{\kappa_{_{_{\mathcal{S}\mathcal{U}}}}}{{D}^\theta z},   &  \text{$0<z<\frac{\kappa_{_{_{\mathcal{S}\mathcal{U}}}}r^\theta}{{\kappa_{_{_{\mathcal{S}\mathcal{A}}}}D}^{\theta}},\ $}\\
\frac{\kappa_{_{_{\mathcal{S}\mathcal{A}}}}}{r^\theta}, & \text{$z \geq \frac{\kappa_{_{_{\mathcal{S}\mathcal{U}}}}r^\theta}{{\kappa_{_{_{\mathcal{S}\mathcal{A}}}}D}^{\theta}}.\ $}
\end{cases}
\end{align}
 Using  integral of incomplete  beta function \cite[eq. (06.19.21.0002.01)]{betaintegral}, $ \int \mathrm{B}_z(a,b) \mathrm{d}z= z \mathrm{B}_z(a,b)- \mathrm{B}_z(a+1,b)$ and further simplifying, we can found   ${{f}}_{_{\bar{\gamma}^{^ {\mathcal{S}\mathcal{U}}}_{_{\mathcal{S}\mathcal{A}}}}}$ as a piece-wise expression given in \eqref{eq:PDFratio}.
\end{proof}
%\vspace{-7mm}
\subsubsection{ $f_{\log_2\left({\bar{\gamma}^{^ {\mathcal{S}\mathcal{U}}}_{_{\mathcal{S}\mathcal{A}}}}\right)} \left(c\right)$} %SC is a logarithmic function of ratio
% %For analytical tractability,  SC may be considered as  logarithmic function of ratio 
% $ {\bar{\gamma}^{^ {\mathcal{S}\mathcal{U}}}_{_{\mathcal{S}\mathcal{A}}}}$ under the cases 1 and 2 as discussed in Section IV.  
To draw critical insights about SC, PDF of logarithmic function of SNRs ratio is given as:
\begin{align}\label{eq:PDFSC1}
%\begin{split}
{{f}}_{\log_2\left({\bar{\gamma}^{^ {\mathcal{S}\mathcal{U}}}_{_{\mathcal{S}\mathcal{A}}}}\right)}\left(c\right) \triangleq \ln({2}){f}_{\bar{\gamma}^{^ {\mathcal{S}\mathcal{U}}}_{_{\mathcal{S}\mathcal{A}}}}\left(2^c\right)
\hspace{-3mm}
%\end{split}
\end{align}
% \begin{align}\label{eq:PDFSC1}
% %\begin{split}
% {{f}}_{\log_2\left({\bar{\gamma}^{^ {\mathcal{S}\mathcal{U}}}_{_{\mathcal{S}\mathcal{A}}}}\right)}\left(c\right)\nonumber \triangleq \ln({2}){f}_{\bar{\gamma}^{^ {\mathcal{S}\mathcal{U}}}_{_{\mathcal{S}\mathcal{A}}}}\left(2^c\right)
% % \begin{cases}
% %   \frac{ \ln({2})(\lambda_{\mathcal{E}})^{-\frac{2}{\theta}}D^4 2^{\frac{2c}{\theta}} }{ F_{d_{_{\mathcal{S}\mathcal{U}}}}(D) \theta R^2r^2} -\frac{ 32\ln({2})({{\lambda_{\mathcal{E}}}})^{-\frac{2}{\theta}}R^2 2^{\frac{2c}{\theta}}\mathrm{\beta}_1\left(\frac{D^2}{4R^2}\right)}{ F_{d_{_{\mathcal{S}\mathcal{U}}}}(D)  \pi \theta r^2}, \\\qquad \qquad \qquad \qquad \qquad \forall\;\text{ $0<c<\log_2(\frac{{\lambda_{\mathcal{E}}}r^\theta}{{D}^{\theta}}),\ $}\\
% %   \frac {\ln({2})(\lambda_{\mathcal{E}})^{\frac{2}{\theta}} r^2c^{-\frac{2}{\theta}}}{  F_{d_{_{\mathcal{S}\mathcal{U}}}}(D)  \theta R^2} -\frac{32\ln({2})({\lambda_{\mathcal{E}}})^{-\frac{2}{\theta}}R^2 2^{\frac{2c}{\theta}}\mathrm{\beta}_1 \left({\frac{({\lambda_{\mathcal{E}}})^{\frac{2}{\theta}} r^2}{2^{\frac{2c}{\theta}}4R^2}}\right) }{ F_{d_{_{\mathcal{S}\mathcal{U}}}}(D)  \pi \theta r^2},  \\\qquad \qquad \qquad \qquad \qquad \forall\;  \text{  $ \log_2(\frac{{\lambda_{\mathcal{E}}}r^\theta}{{D}^{\theta}})\ \leq c <\infty.  $}
% %   %%(1-\alpha)\delta(z-\infty), & \text{$ z=\infty $} \\
% %   \hspace{-5mm}
% % \end{cases}
% \hspace{-3mm}
% %\end{split}
% \end{align}
%%%%%%%%%%%%%%% lemma 2 jamming Pdf of snr ratio
\subsubsection{${{f}}_{_{\bar{\gamma}^{^ {\mathcal{S}\mathcal{U}}}_{_{\mathcal{A}\mathcal{U}}}}}\left(z\right)$}  %PDF of ratio ${\bar{\gamma}^{^ {\mathcal{S}\mathcal{U}}}_{_{\mathcal{A}\mathcal{U}}}}$ is provided by Lemma 2.
Lemma 2 provides the PDF of SNRs ratio under jamming.
\begin{lemma}\label{lemma:Lemma2}
%\italemma{lemma 1 b}
The PDF ${{f}}_{_{\bar{\gamma}^{^ {\mathcal{S}\mathcal{U}}}_{_{\mathcal{A}\mathcal{U}}}}}(z)$   of  legitimate-to-jamming link SNRs ratio ${\bar{\gamma}^{^ {\mathcal{S}\mathcal{U}}}_{_{\mathcal{A}\mathcal{U}}}}$ is given below.
\begin{align}\label{eq:PDFratio_jamming}
%\begin{split}
{f}_{_{\bar{\gamma}^{^ {\mathcal{S}\mathcal{U}}}_{_{\mathcal{A}\mathcal{U}}}}}\left(z\right)\triangleq \begin{cases}
  \frac{ (\lambda_\mathcal{J})^{-\frac{2}{\theta}}D^4 z^{\frac{2}{\theta}-1}}{ F_{d_{_{\mathcal{S}\mathcal{U}}}}(D)R^4} -\frac{ 32({\lambda_\mathcal{J}})^{-\frac{2}{\theta}}z^{\frac{2}{\theta}-1}\mathrm{\beta}_1\left(\frac{D^2}{4R^2}\right)}{ F_{d_{_{\mathcal{S}\mathcal{U}}}}(D)},\\\qquad \qquad \qquad \qquad \qquad \forall\;
    \text{  $0<z<\frac{\kappa_{_{_{\mathcal{S}\mathcal{U}}}}R^\theta}{{\kappa_{_{_{\mathcal{S}\mathcal{A}}}}D}^{\theta}},\ $}\\
\frac {({\lambda_{\mathcal{J}}})^{\frac{2}{\theta}} z^{-\frac{2}{\theta}-1}}{  F_{d_{_{\mathcal{S}\mathcal{U}}}}(D)  \theta } -\frac{32({\lambda_{\mathcal{J}}})^{-\frac{2}{\theta}} z^{\frac{2}{\theta}-1}\mathrm{\beta}_1 \left({\frac{({\lambda_{\mathcal{J}}})^{\frac{2}{\theta}} }{4z^{\frac{2}{\theta}}}}\right) }{ F_{d_{_{\mathcal{S}\mathcal{U}}}}(D)  \pi \theta },\\ \text{  $ \qquad \qquad \qquad \qquad \qquad \forall\; \frac{\kappa_{_{_{\mathcal{S}\mathcal{U}}}}R^\theta}{{\kappa_{_{_{\mathcal{S}\mathcal{A}}}}D}^{\theta}}\ \leq  z <\infty,   $}
\end{cases}
%\end{split}
\end{align} 
where $ {\lambda_\mathcal{J}}=\left(\frac{\kappa_{_{_{\mathcal{S}\mathcal{U}}}}}{\kappa_{_{_{\mathcal{A}\mathcal{U}}}}}\right)$.
\end{lemma}
 \begin{proof}
%  In case of jamming attack, $ P_{\mathcal{J}} $ is controlling parameter instead of $r$. Therefore we have considered that \mathcal{U} is present in deployment region $R $ with a probability 
%$Here,$ {\gamma_{_{\mathcal{A}\mathcal{U}}}}$ is a not  mixture random variable.
 PDF of ${\bar{\gamma}^{^ {\mathcal{S}\mathcal{U}}}_{_{\mathcal{A}\mathcal{U}}}}$  can be derived similar to Lemma~\ref{lemma: lemma1} by using
 $f_{\gamma_{_{\mathcal{S}\mathcal{U}}}}\left(x\right)$  and  $f_{\gamma_{_{\mathcal{A}\mathcal{U}}}}\left(y\right )$ defined in  \eqref{eq:PDFSNRSU} and  \eqref{eq:PDFSNRAU} respectively. The only difference lies in case of jamming that, $\frac{\gamma_{_{\mathcal{S}\mathcal{U}}}}{{\gamma_{_{\mathcal{A}\mathcal{U}}}}}$ exists for entire region of interest of radius $R$. 
\end{proof}
\subsubsection{ $f_{\log_2\left(1+{\bar{\gamma}^{^ {\mathcal{S}\mathcal{U}}}_{_{\mathcal{A}\mathcal{U}}}}\right)} \left(c\right)$} Under jamming mode, SC depends upon the  logarithmic function of 
%For analytical tractability,  SC may be considered as  logarithmic function of ratio 
$ (1+{\bar{\gamma}^{^ {\mathcal{S}\mathcal{U}}}_{_{\mathcal{A}\mathcal{U}}}})$ under high SNR regime as discussed in Section 4.  Therefore,  PDF of  ${\log_2\left(1+{\bar{\gamma}^{^ {\mathcal{S}\mathcal{U}}}_{_{\mathcal{A}\mathcal{U}}}}\right)} $ is derived to provide analytical insights about high SNR SC.
\begin{align}\label{eq:PDFratio_jamming_SC}
\begin{split}
{{f}}_{\log_2\left(1+{\bar{\gamma}^{^ {\mathcal{S}\mathcal{U}}}_{_{\mathcal{A}\mathcal{U}}}}\right)}\left(c\right) 
\triangleq  \ln(2){f}_{_{\bar{\gamma}^{^ {\mathcal{S}\mathcal{U}}}_{_{\mathcal{A}\mathcal{U}}}}}\left(2^c-1\right) 
% &\begin{cases}
%   \frac{\ln(2) 2^c(\lambda_\mathcal{J})^{-\frac{2}{\theta}}D^4 {(2^c-1)}^{\frac{2}{\theta}-1}}{ F_{d_{_{\mathcal{S}\mathcal{U}}}}(D)R^4}\\ -\frac{ 32\ln(2) 2^c({\lambda_\mathcal{J}})^{-\frac{2}{\theta}}{(2^c-1)}^{\frac{2}{\theta}-1}\mathrm{\beta}_1 \left(\frac{D^2}{4R^2}\right)}{ F_{d_{_{\mathcal{S}\mathcal{U}}}}(D)},\\\qquad \qquad \forall\; \text{  $ 0<c<\left(1+\frac{\kappa_{_{_{\mathcal{S}\mathcal{U}}}}R^\theta}{{\kappa_{_{_{\mathcal{S}\mathcal{A}}}}D}^{\theta}}\right)\ $,}\\
% \frac {\ln(2) 2^c({\lambda_{\mathcal{J}}})^{\frac{2}{\theta}} {(2^c-1)}^{-\frac{2}{\theta}-1}}{  F_{d_{_{\mathcal{S}\mathcal{U}}}}(D)  \theta }\\ -\frac{32\ln(2) 2^c({\lambda_{\mathcal{J}}})^{-\frac{2}{\theta}} {(2^c-1)}^{\frac{2}{\theta}-1}\mathrm{\beta}_1 \left({\frac{({\lambda_{\mathcal{J}}})^{\frac{2}{\theta}} }{4{(2^c-1)}^{\frac{2}{\theta}}}}\right) }{ F_{d_{_{\mathcal{S}\mathcal{U}}}}(D)  \pi \theta},\\ \qquad \qquad \forall\;\text{ $ \left(1+ \frac{\kappa_{_{_{\mathcal{S}\mathcal{U}}}}R^\theta}{{\kappa_{_{_{\mathcal{S}\mathcal{A}}}}D}^{\theta}}\right)\ \leq c <\infty.$}
% \end{cases}
\end{split}
\end{align} 

\section{Secrecy Performance Analysis }\label{sec:SecracyRate} 
In this section,  the impact of random inter-node distances on the secrecy performance  of the proposed system is presented by adopting  secrecy  outage probability SOP  as a performance metric. Firstly,  SOP for the proposed system model is derived by considering both modes eavesdropping as well as jamming  of  the $\mathcal{A}$.  Later,  the proposed analysis is extended in high SNR regime to obtain the closed-form expressions for the SOP.
\subsection{Secrecy Outage Probability}
 Secrecy outage is  a state, when instantaneous SC  falls below a target  secrecy rate $ C_{st}>0 $. The probability of secrecy outage, SOP,  is defined \cite{Bloch} as: 
\begin{align}\label{eq:outage_Exact}
{p}^i_{o}(C_{st})=\text{Pr}(C^i<C_{st}) \qquad{i}\in\left\lbrace\mathcal{E},\mathcal{J}\right\rbrace,
\end{align}
where $\mathcal{E}$ and $\mathcal{J}$ represents eavesdropping and jamming mode of $\mathcal{A}$ respectively.
\subsubsection{Eavesdropping}
In the proposed model,  under the cases 1 and 2 when  $ \mathcal{S}$  is deployed inside the eavesdropping zone as followed from Section 2, SC    will follow the definition provided by \eqref{eq: C_E}.  It is denoted by  $C^\mathcal{E}_{_{1,2}}$.  On the other hand, under the cases 3 and  4, when  $ \mathcal{S}$ may remain outside the eavesdropping zone,  $\mathcal{A}$ is not able to eavesdrop the signal. Hence, it  becomes  ineffective. Consequently, conventional SC given by \eqref{eq: C_E} is reduced to capacity of legitimate channel and denoted by $C^\mathcal{E}_{_{3,4}}$. Thus,  SC for the proposed system under eavesdropping mode of $\mathcal{A}$ is defined as:  
 \begin{align}\label{eq:C_e_exact_total}
\text{{${C}^{\mathcal{E}} \triangleq \begin{cases} 
C^\mathcal{E}_{_{1,2}}=\left[\log_{2}\left(\frac{ {1+\gamma_{_{\mathcal{S}\mathcal{U}}}}}{ {1+\gamma_{_{\mathcal{S}\mathcal{A}}}}}\right )\right]^+, & \text{$\gamma_{_{\mathcal{S}\mathcal{A}}}>\frac{\kappa_{_{_{\mathcal{S}\mathcal{A}}}}}{{(r)}^{\theta}},\ $}\\
 C^\mathcal{E}_{_{3,4}}=\log_2 \left(1+\gamma_{_{\mathcal{S}\mathcal{U}}}\right),   & \text{otherwise.}
\end{cases}$}}
\end{align} 
Consequently,  the  SOP under eavesdropping mode of $\mathcal{A}$ will be the weighted sum of two outage probabilities   given below.
\begin{align}\label{eq:outage_weighted}
{p}_{o}^{\mathcal{E}}=\alpha p^E_{o}+(1-\alpha)p^I_{o},
\end{align} 
 where $ p^E_{o}= \text{Pr}\left(\log_2\left(\frac{1+\gamma_{_{\mathcal{S}\mathcal{U}}}}{1+\gamma_{_{\mathcal{S}\mathcal{A}}}}\right)<C_{st}\right)$  represents outage probability  when eavesdropper is effective i.e., the cases 1 and 2 and $ p^I_{o}=\text{Pr}(\log_2({1+\gamma_{_{\mathcal{S}\mathcal{U}}}})<C_{st})$ represents outage probability when  eavesdropper becomes ineffective  i.e., the cases 3 and 4. The ${p}_{o}^{\mathcal{E}}$ can be derived as:
\begin{align}\label{eq:outage_E_Th1}
   &p_o^{\mathcal{E}}= \alpha\text{Pr}\left(\frac{1+\gamma_{_{\mathcal{S}\mathcal{U}}}}{1+\gamma_{_{\mathcal{S}\mathcal{A}}}}<2^{C_{st}}\right)\nonumber
  + (1-\alpha)   \text{Pr}\left({1+\gamma_{_{\mathcal{S}\mathcal{U}}}}<2^{C_{st}}\right)\nonumber\\
 & = \alpha \int \limits_0^{2^{C_{st}}} f_{_{\frac{1+\gamma_{_{\mathcal{S}\mathcal{U}}}}{1+\gamma_{_{\mathcal{S}\mathcal{A}}}}}}\left(z\right)\mathrm{d}z+ (1-\alpha) \int \limits_0^{2^{C_{st}}} f_{_{{1+\gamma_{_{\mathcal{S}\mathcal{U}}}}}}\left(z\right)\mathrm{d}z \nonumber\\
   &\stackrel{(b)}{=} 
   \int\limits_0^{2^{C_{st}}} \left(\alpha\int\limits_{y_{l}}^{\infty} yf_{(1+\gamma_{_{\mathcal{S}\mathcal{U}}})}(zy)f_{(1+\gamma_{_{\mathcal{S}\mathcal{A}}})}((y|y>\frac{\kappa_{_{_{\mathcal{S}\mathcal{A}}}}}{{r}^{\theta}}+1))\mathrm{d}y\right)\mathrm{d}z \nonumber\\&
   +(1-\alpha)\int\limits_0^{2^{C_{st}}}f_{(1+\gamma_{_{\mathcal{S}\mathcal{U}}})}\mathrm{d}z\nonumber\\
  &\stackrel{(c)}{=}\alpha \int\limits_0^{2^{C_{st}}} \int\limits_{y_l}^{\infty}\frac{ 16u_1y{(y-1)}^{-\frac{2}{\theta}-1}\left(\pi-2 B_{u_1}\left(\frac{1}{2},\frac{3}{2}\right)\right)}{ F_{d_{_{\mathcal{S}\mathcal{U}}}}(D) \pi (z y-1)r^2 \theta^2 {\kappa_{_{_{\mathcal{S}\mathcal{A}}}}}^{-\frac{2}{\theta}}}\mathrm{d}y \mathrm{d}z \nonumber  \\&+{(1-\alpha)} \left(\frac{D^2-u_2}{F_{d_{_{\mathcal{S}\mathcal{U}}}}(D) R^2}-\frac{8\left(\beta_1\left(\frac{x_1}{4R^2}\right)-\beta_1\left(\frac{D^2}{4R^2}\right)\right)}{F_{d_{_{\mathcal{S}\mathcal{U}}}}(D) \pi}\right),  
    \end{align} 
Here, $(b)$ is obtained by \cite[(eq. 6.60)]{populis}  for the cases 1 and 2 whereas   $(c )$ is obtained by finding the PDF $f_{(1+\gamma_{_{\mathcal{S}\mathcal{U}}})}$ and  $f_{(1+\gamma_{_{\mathcal{S}\mathcal{A}}})}$  using random variable transformations on  \eqref{eq:PDFSNRSU}, and \eqref{eq:PDFSNRSA}. The $ y_{_l}$ is lower limits of internal integral which is maximum of lower supports of $(1+{\gamma_{_{\mathcal{S}\mathcal{U}}}})$ and $ ({1+{\gamma_{_{\mathcal{S}\mathcal{A}}}}})$ Hence, $ y_l $ is  $\max\left (\frac{{D}^\theta+\kappa_{_{_{\mathcal{S}\mathcal{U}}}}}{{D}^\theta z},\frac{\kappa_{_{_{\mathcal{S}\mathcal{A}}}}+r^\theta}{{r}^\theta}\right)$ and can be given in piece-wise form as follows:
\begin{align}
 {y}_l=\begin{cases}
 \frac{\kappa_{_{_{\mathcal{S}\mathcal{U}}}}}{{D}^\theta z},   &  \text{$0<z<\frac{\kappa_{_{_{\mathcal{S}\mathcal{U}}}}r^\theta}{{\kappa_{_{_{\mathcal{S}\mathcal{A}}}}D}^{\theta}},\ $}\\
\frac{\kappa_{_{_{\mathcal{S}\mathcal{A}}}}}{r^\theta}, & \text{$z \geq \frac{\kappa_{_{_{\mathcal{S}\mathcal{U}}}}r^\theta}{{\kappa_{_{_{\mathcal{S}\mathcal{A}}}}D}^{\theta}}.\ $}
\end{cases}
\end{align}
{The  first term in $(c)$ represents  integral expression for $p_o^E$.  Though it cannot be solved analytically due to involvement of the form $\int (A+Bx)^n(C+Dx)^m\mathrm{d}x$ \cite[(2.151)]{Gradshteyn}, it can be solved numerically.}
% The second term of $(c)$ represents $p_o^I$ which is obtained through solving integral in second term of  $(b)$ by using \cite[eq. (06.19.21.0002.01)]{betaintegral} and algebraic calculations.
It is to be mentioned here that the integral in the second term of  $(b)$ represents $p_o^I$ which is solved in $(c)$ by using \cite[eq. (06.19.21.0002.01)]{betaintegral} and algebraic calculations.
\subsubsection{Jamming}
For the proposed system, the SOP $p_o^{\mathcal{J}}$  under jamming  mode of $\mathcal{A}$ can be given using the equation \eqref{eq: C_J} and  \eqref{eq:outage_Exact} as follows:
\begin{align}\label{eq:outage_J_exact}
p_o^{\mathcal{J}} = &\text{Pr}\left(\log_2\left(1+\frac{\gamma_{_{\mathcal{S}\mathcal{U}}}}{1+\gamma_{_{\mathcal{S}\mathcal{A}}}}\right)<C_{st}\right)\nonumber\\&=\text{Pr}\left(\left(\frac{\gamma_{_{\mathcal{S}\mathcal{U}}}}{1+\gamma_{_{\mathcal{S}\mathcal{A}}}}\right)<2^{C_{st}}-1\right)=
 \int \limits_0^{2^{C_{st}}-1} f_{_{\frac{\gamma_{_{\mathcal{S}\mathcal{U}}}}{1+\gamma_{_{\mathcal{S}\mathcal{A}}}}}}\left(z\right)\mathrm{d}z \nonumber\\
 &\stackrel{(d)}{=} \displaystyle {\int\limits_0^{2^{C_{st}}-1} \int\limits_{y_{l}}^{\infty} yf_{\gamma_{_{\mathcal{S}\mathcal{U}}}}(zy)f_{(1+\gamma_{_{\mathcal{S}\mathcal{A}}})}(y)\mathrm{d}y \mathrm{d}z\nonumber}\\ &\stackrel{(e)}{=} \displaystyle
  \int\limits_0^{2^{C_{st}}-1} \int\limits_{y_l}^{\infty}\frac{ 16u_1{(y-1)}^{-\frac{2}{\theta}}\left(\pi-2 B_{u_1}\left(\frac{1}{2},\frac{3}{2}\right)\right)}{ F_{d_{_{\mathcal{S}\mathcal{U}}}}(D) \pi (z y) R^2 \theta^2 {\kappa_{_{_{\mathcal{S}\mathcal{A}}}}}^{-\frac{2}{\theta}}}\mathrm{d}y \mathrm{d}z,
\end{align}
where, $(d)$ is obtained by \cite[(eq. 6.60)]{populis} and   $(e )$ is obtained by using \eqref{eq:PDFSNRSU} and finding the PDF   $f_{(1+\gamma_{_{\mathcal{S}\mathcal{A}}})}$  using random variable transformations on  \eqref{eq:PDFSNRSA}. The $ y_{_l}$ is lower limits of internal integral which is  maximum of lower supports of ${\gamma_{_{\mathcal{S}\mathcal{U}}}}$ and $ ({1+{\gamma_{_{\mathcal{S}\mathcal{A}}}}})$ hence, $ y_l $ is  $\max\left (\frac{\kappa_{_{_{\mathcal{S}\mathcal{U}}}}}{{D}^\theta z},\frac{\kappa_{_{_{\mathcal{S}\mathcal{A}}}}+R^\theta}{{r}^\theta}\right)$ and can be given in piece-wise form as:
 \begin{align} {y}_l=\begin{cases}
 \frac{\kappa_{_{_{\mathcal{S}\mathcal{U}}}}}{{D}^\theta z},   &  \text{$0<z<\frac{\kappa_{_{_{\mathcal{S}\mathcal{U}}}}R^\theta}{\left({\kappa_{_{_{\mathcal{S}\mathcal{A}}}}+R^\theta}\right)D^{\theta}},\ $}\\
\frac{R^\theta + \kappa_{_{_{\mathcal{S}\mathcal{A}}}}}{R^\theta}, & \text{$z \geq \frac{\kappa_{_{_{\mathcal{S}\mathcal{U}}}}R^\theta}{\left({\kappa_{_{_{\mathcal{S}\mathcal{A}}}}+R^\theta}\right)D^{\theta}}.\ $}\\
\end{cases}   
\end{align}
\subsection{Closed-Form Approximation}
The  analysis presented above provides the integral-based SOP expressions. To provide additional analytical insights, we consider high SNR regime and deduce closed-form expressions for SOP for  tractable analytical results.
This analysis also  provides lower bound on SOP in jamming whereas under eavesdropping, lower bound on SOP is provided with a condition of  positive secrecy.
\subsubsection{Eavesdropping}
For high SNR regime, received signals in the $\mathcal{S}$-$\mathcal{U}$ and $\mathcal{S}$-$\mathcal{A}$ links are relatively higher  than noise power. Hence, $\gamma_{_{\mathcal{S}\mathcal{U}}}\gg1$ and $\gamma_{_{\mathcal{S}\mathcal{A}}}\gg1$ are satisfied when $\mathcal{A}$ is able to eavesdrop the legitimate signal i.e., the cases 1 and 2 for the proposed model and corresponding SC is denoted as ${\hat{C}}^{\mathcal{E}}_{_{1,2}}$. Therefore, SC for proposed system \eqref{eq:C_e_exact_total} is reduced to asymptotic SC defined as:
 \begin{align}\label{eq:C_e_hi_total}
\text{{${\hat{C}}^{\mathcal{E}} \triangleq \begin{cases} {\hat{C}}^{\mathcal{E}}_{_{1,2}}=
\left[\log_{2}\left(\frac{ {\gamma_{_{\mathcal{S}\mathcal{U}}}}}{ {\gamma_{_{\mathcal{S}\mathcal{A}}}}}\right )\right]^+, & \text{$\gamma_{_{\mathcal{S}\mathcal{A}}}>\frac{\kappa_{_{_{\mathcal{S}\mathcal{A}}}}}{{r}^{\theta}},\ $}\\
 {{C}}^{\mathcal{E}}_{_{3,4}}=\log_2 \left(1+\gamma_{_{\mathcal{S}\mathcal{U}}}\right),   & \text{otherwise.}
\end{cases}$}}
\end{align} 
   It is to be mentioned that  SOP expression derived for SOP considering the cases 3 and 4 is analytically tractable. The respective exact expression is, therefore, retained  aiming  to enhance the accuracy of closed-form solutions.
 
 For a special case of positive SC i.e., ${\gamma_{_{\mathcal{S}\mathcal{U}}}}>  {\gamma_{_{\mathcal{S}\mathcal{A}}}}$,  the proposed asymptotic SC gives the upper bound on SC given by \eqref{eq:C_e_exact_total} as $ C^{{\mathcal{E}}}   \leq  {\hat{C}}^{\mathcal{E}}  $.  
   Consequently, the corresponding SOP  $\hat{p}^\mathcal{E}_{o}$ provides lower bound on outage probability with positive SC.   Theorem 1 is introduced to provide closed-form expression for SOP $\hat{p}^\mathcal{E}_{o}$.
\begin{theorem}
Under eavesdropping attack, SOP $\hat{p}^\mathcal{E}_{o}$ for a target secrecy rate $C_{st}>0 $ is given by:  \begin{align}\label{eq:outage_E_Th2_highsnr}
\begin{split}
 \hat{p}^\mathcal{E}_{o} \triangleq
 \begin{cases}
 \alpha\left(\frac{ (\lambda_{\mathcal{E}})^{-\frac{2}{\theta}}D^4 2^{\frac{2 C_{st}}{\theta}}}{ 2 F_{d_{_{\mathcal{S}\mathcal{U}}}}(D)  R^2r^2} -\frac{ 32(\lambda_{\mathcal{E}})^{-\frac{2}{\theta}}R^2 2^{\frac{2 C_{st}}{\theta}}\beta_1(\frac{D^2}{4R^2})}{  2 F_{d_{_{\mathcal{S}\mathcal{U}}}}(D)   r^2} \right)\\
 +{(1-\alpha)} \left(\frac{D^2-u_2}{F_{d_{_{\mathcal{S}\mathcal{U}}}}(D) R^2}-\frac{8\left(\beta_1\left(\frac{u_2}{  4R^2}\right)-\beta_1\left(\frac{D^2}{4R^2}\right)\right)}{ F_{d_{_{\mathcal{S}\mathcal{U}}}}(D) \pi}\right),\\  \qquad \qquad \qquad \qquad \qquad 
 \text{$ \forall\; 0<C_{st}< \log_2 \left(\frac{\lambda_{\mathcal{E}} r^\theta}{D^{\theta}}\right)\ $}\\
{\alpha} \left( \frac{4D^2-u_3}{2R^2F_{d_{_{\mathcal{S}\mathcal{U}}}}(D) }-\frac{16 R^2\beta_2(\frac{D^2}{4R^2})}{\pi F_{d_{_{\mathcal{S}\mathcal{U}}}}(D)  D^2 } + \frac{4(\beta_3(\frac{u_3}{4R^2})-\beta_3(\frac{D^2}{4R^2}))}{\pi F_{d_{_{\mathcal{S}\mathcal{U}}}}(D) } \right)\\
 + {(1-\alpha)} \left(\frac{D^2-u_2}{F_{d_{_{\mathcal{S}\mathcal{U}}}}(D) R^2}-\frac{8\left(\beta_1\left(\frac{u_2}{  4R^2}\right)-\beta_1\left(\frac{D^2}{4R^2}\right)\right)}{ F_{d_{_{\mathcal{S}\mathcal{U}}}}(D) \pi}\right),\\ \qquad \qquad \qquad \qquad \qquad \text{  $ \forall \;\log_2 \left (\frac{\lambda_{\mathcal{E}} r^\theta}{D^{\theta}}\right) \leq C_{st} $},  \\
\end{cases}
 \end{split}
\end{align} 
% \normalsize
 with $u_3=\frac{r^2  (\lambda_{\mathcal{E}})^  {\frac{2}{\theta}}2^{\frac{-2C_{st}}{\theta}}}{2R^2}$, $\beta_2(a)= a \mathrm{B}_a(\frac{1}{2},\frac{3}{2})- \mathrm{B}_a(\frac{5}{2},\frac{3}{2}) $, and 
$\beta_3(a)= \beta_1(a)+ \frac{1}{a} \mathrm{B}_a(\frac{5}{2},\frac{3}{2})- \mathrm{B}_a(\frac{3}{2},\frac{3}{2}).$ 
\end{theorem}
\begin{proof} It is clear from \eqref{eq:outage_weighted} that SOP  $\hat{p}^\mathcal{E}_{o}$ is weighted sum of two outages probabilities $p_o^E$ (for cases 1 and 2) and $p_o^I$ ( for the cases 3 and 4) obtained as below.\\
Cases 1 and 2: $\mathcal{S}$ lies within  eavesdropping zone of $\mathcal{A}$ with  probability $\alpha$, and  $p_o^E$ is defined as:
\begin{align}\label{eq:outage_E_Th1_proof}
&p_o^E \approx \text{Pr}\left (\log_2\left({\bar{\gamma}^{^ {\mathcal{S}\mathcal{U}}}_{_{\mathcal{S}\mathcal{A}}}}\right)<C_{st}\right)\nonumber\\&= \text{Pr}\left ({\bar{\gamma}^{^ {\mathcal{S}\mathcal{U}}}_{_{\mathcal{S}\mathcal{A}}}}<2^{C_{st}}\right)\nonumber =  \int\limits_0^{2^{C_{st}}} f_{_{{\bar{\gamma}^{^ {\mathcal{S}\mathcal{U}}}_{_{\mathcal{S}\mathcal{A}}}}}}\left(z\right)\mathrm{d}z \nonumber\\
   &\stackrel{(f)}{=}%\hspace{-2mm}
   \int\limits_0^\frac{\kappa_{_{_{\mathcal{S}\mathcal{U}}}}r^\theta}{{\kappa_{_{_{\mathcal{S}\mathcal{A}}}}D}^{\theta}} \hspace{-1.5mm} \left( \frac{ (\lambda_\mathcal{E})^{-\frac{2}{\theta}}D^4 z^{\frac{2}{\theta}-1}}{ F_{d_{_{\mathcal{S}\mathcal{U}}}}(D) R^2r^2} \right)\mathrm{d}z+ \int\limits_\frac{\kappa_{_{_{\mathcal{S}\mathcal{U}}}}r^\theta}{{\kappa_{_{_{\mathcal{S}\mathcal{A}}}}D}^{\theta}}^ {2^{C_{st} }}\hspace{-2mm} \left(\frac {(\lambda_\mathcal{E})^{\frac{2}{\theta}} r^2}{  F_{d_{_{\mathcal{S}\mathcal{U}}}}(D) R^2}\nonumber\right)\mathrm{d}z\nonumber\\&
   - \int\limits_0^\frac{\kappa_{_{_{\mathcal{S}\mathcal{U}}}}r^\theta}{{\kappa_{_{_{\mathcal{S}\mathcal{A}}}}D}^{\theta}}\left(\frac{ 32(\lambda_\mathcal{E})^{-\frac{2}{\theta}}R^2z^{\frac{2}{\theta}-1}\beta_1(\frac{D^2}{4R^2})}{ F_{d_{_{\mathcal{S}\mathcal{U}}}}(D) r^2} \right)\mathrm{d}z\nonumber \\&
%   +%\hspace{-2.5mm}
%   \int\limits_\frac{\kappa_{_{_{\mathcal{S}\mathcal{U}}}}r^\theta}{{\kappa_{_{_{\mathcal{S}\mathcal{A}}}}D}^{\theta}}^ {2^{C_{st} }}\hspace{-2mm} \left(\frac {(\lambda_\mathcal{E})^{\frac{2}{\theta}} r^2}{  F_{d_{_{\mathcal{S}\mathcal{U}}}}(D) R^2}\nonumber\right)\\&
+\int\limits_\frac{\kappa_{_{_{\mathcal{S}\mathcal{U}}}}r^\theta}{{\kappa_{_{_{\mathcal{S}\mathcal{A}}}}D}^{\theta}}^ {2^{C_{st} }}\left(\frac{32(\lambda_\mathcal{E})^{-\frac{2}{\theta}}R^2z^{\frac{2}{\theta}-1}\mathrm{\beta}_1 \left({\frac{({\lambda_{\mathcal{E}}})^{\frac{2}{\theta}} r^2}{z^{\frac{2}{\theta}}4R^2}}\right)}{ F_{d_{_{\mathcal{S}\mathcal{U}}}}(D) r^2} \right)\mathrm{d}z
 \end{align}
 %\vspace{-2mm}
Here, $(f)$ is obtained  using Lemma 1.\\ 
Cases 3 and 4: $\mathcal{S}$ remains beyond the eavesdropping zone  with $(1-\alpha)$ probability.  Under these cases, we consider $p_o^I$ remains the same as in  \eqref{eq:outage_E_Th1}. 

Solving  $p_o^E$, \cite[eq. (06.19.21.0002.01)]{betaintegral} and putting values of  $p_o^I$ and  $p_o^E$ in \eqref{eq:outage_weighted}, a closed-form expression for SOP   in  \eqref{eq:outage_E_Th2_highsnr} is obtained. \end{proof}
%\end{theorem}
\subsubsection{Jamming} In high SNR regime,  interference produced by jamming power has a much greater impact than noise in the main channel. Hence, for this interference limited case,  the SC given by \eqref{eq: C_J}  is reduced to asymptotic SC denoted as follows:
 {  \begin{align}\label{eq:C_J_asym2} {\hat{C}}^{\mathcal{J}} & \triangleq
  \log_{2}\left(1+ \frac{\gamma_{_{\mathcal{S}\mathcal{U}}}}{ \gamma_{_{\mathcal{A}\mathcal{U}}}}\right)\end{align}
     }
The $ {\hat{C}}^{\mathcal{J}}$ also provides the upper bound on SC given by \eqref{eq: C_J} as  ${C}^{\mathcal{J}}  \leq  {\hat{C}}^{\mathcal{J}}$. Consequently, asymptotic SOP $\hat{{p}}^\mathcal{J}_{o}$ corresponding to $ {\hat{C}}^{\mathcal{J}}$ provides  the lower bound on SOP  $\hat{{p}}^\mathcal{J}_{o}$ in closed-form as given by Theorem 2.
\begin{theorem} Under jamming attack, SOP for a target secrecy rate $C_{st} $  is lower bounded by: 
\begin{align}\label{eq:outage_J_Th3}
\begin{split}
\hat{p}^\mathcal{J}_{o} \triangleq
\begin{cases}
\frac{ \lambda_{\mathcal{J}}^{-\frac{2}{\theta}}D^4 (2^{C_{st}}-1)^{\frac{2 }{\theta}}}{ 2 F_{d_{_{\mathcal{S}\mathcal{U}}}}(D)  R^4} -\frac{ 32(\lambda_{\mathcal{J}})^{-\frac{2}{\theta}} (2^{C_{st}}-1)^{\frac{2 }{\theta}}\beta_1(\frac{D^2}{4R^2})}{  2 F_{d_{_{\mathcal{S}\mathcal{U}}}}(D)},\\  \qquad \qquad \text{ $ \forall \;0<C_{st}< \log_2 \left(1+\frac{\lambda_{\mathcal{E}} r^\theta}{D^{\theta}}\right)\ $}\\
  \frac{D^2-4u_4}{R^2F_{d_{_{\mathcal{S}\mathcal{U}}}}(D) }-\frac{16 R^2\beta_2(\frac{D^2}{4R^2})}{\pi F_{d_{_{\mathcal{S}\mathcal{U}}}}(D)  D^2 } + \frac{4(\beta_3({u_4})-\beta_3(\frac{D^2}{4R^2}))}{\pi F_{d_{_{\mathcal{S}\mathcal{U}}}}(D) },\\  \qquad \text{  $ \forall \; \log_2 \left (1+\frac{\lambda_{\mathcal{E}} r^\theta}{D^{\theta}}\right) \leq C_{st} $ where $u_4={\frac{(\lambda_\mathcal{J})^{\frac{2}{\theta}}}{4(2^{C_{st}}-1)^{\frac{2 }{\theta}}}}.$} \\
\end{cases}
\end{split}
\end{align} 
% with $u_4={\frac{(\lambda_\mathcal{J})^{\frac{2}{\theta}}}{4(2^{C_{st}}-1)^{\frac{2 }{\theta}}}}.$
  \end{theorem}
   \begin{proof}  By using \eqref{eq:outage_Exact} and  \eqref{eq:C_J_asym2}, we obtain:
  \begin{align}
  %\label{eq:outage_J_proof}
 &\hat{p_{o}}^\mathcal{J} 
 =\text{Pr}\left(\log_2\left(1+{\bar{\gamma}^{^ {\mathcal{S}\mathcal{U}}}_{_{\mathcal{A}\mathcal{U}}}}\right)<C_{st}\right)= \int\limits_0^{2^{C_{st}}-1} f_{_{{\bar{\gamma}^{^ {\mathcal{S}\mathcal{U}}}_{_{\mathcal{A}\mathcal{U}}}}}}\left(z \right)\mathrm{d}z
 \nonumber\\
 &\hspace{-1mm}\stackrel{\left(g\right)}{=}\hspace{-2mm} \int\limits_0^\frac{\kappa_{_{_{\mathcal{S}\mathcal{U}}}}r^\theta}{{\kappa_{_{_{\mathcal{S}\mathcal{A}}}}D}^{\theta}}  \left( \frac{ (\lambda_\mathcal{J})^{-\frac{2}{\theta}}D^4 z^{\frac{2}{\theta}-1}}{ F_{d_{_{\mathcal{S}\mathcal{U}}}}(D)R^4} -\frac{ 32({\lambda_\mathcal{J}})^{-\frac{2}{\theta}}z^{\frac{2}{\theta}-1}\mathrm{\beta}_1(\frac{D^2}{4R^2})}{ F_{d_{_{\mathcal{S}\mathcal{U}}}}(D)}\right) \mathrm{d}z\nonumber\\&
  \hspace{-1mm}+\hspace{-2mm}\int\limits_\frac{\kappa_{_{_{\mathcal{S}\mathcal{U}}}}r^\theta}{{\kappa_{_{_{\mathcal{S}\mathcal{A}}}}D}^{\theta}}^{2^{C_{st}}-1}  \left(\frac {({\lambda_{\mathcal{J}}})^{\frac{2}{\theta}} z^{-\frac{2}{\theta}-1}}{  F_{d_{_{\mathcal{S}\mathcal{U}}}}(D)  \theta } -\frac{32({\lambda_{\mathcal{J}}})^{-\frac{2}{\theta}} z^{\frac{2}{\theta}-1}\mathrm{\beta}_1 \left({\frac{({\lambda_{\mathcal{J}}})^{\frac{2}{\theta}} }{z^{4\frac{2}{\theta}}}}\right) }{ F_{d_{_{\mathcal{S}\mathcal{U}}}}(D)  \pi \theta }\right)\mathrm{d}z
\end{align}
Here, $\left(g\right)$ is obtained with use of Lemma 2.  Using \cite[eq. (06.19.21.0002.01)]{betaintegral} and algebraic manipulation, closed-form  asymptotic SOP $\hat{p_{o}}^\mathcal{J}$   can be  obtained as given in \eqref{eq:outage_J_Th3}.
\end{proof}
\section{Discussion: Model Generalisation and Parameter Designing}
In this section, a general discussion about other spatial configuration of nodes is presented, and it is illustrated that the proposed analysis represents a generalised framework as the corresponding secrecy analysis can be derived in a similar manner. 
In addition, we discuss the impact of the system design parameters of legitimate nodes as well as  attacker on secrecy performance. We observe that determining the appropriate values  of parameters utilising the analytical expressions derived in Section 4 offers the desired secrecy performance for legitimate nodes and attacker. 
\subsection{Generalisation of the proposed System Model}
As discussed in Section 4, analytically tractable secrecy outage analysis can be performed through having determined the  ratio distribution of SNRs of legitimate to attacker link. In this section, we illustrate that ratio distribution of SNRs of legitimate to attacker link can be derived for other possible configurations on the similar lines. We examine two prominent configurations considered in the literature, however other possible configurations can also be analysed with the proposed generalised framework.
            \subsubsection{$\mathcal{S}$ at the origin}
    In this configuration, $\mathcal{A}$ and  $\mathcal{U}$ are located stochastically  within the deployment region  with $\mathcal{S}$  at the centre. Distances  ${d_{_{\mathcal{S}\mathcal{U}}}}$ and ${d_{_{\mathcal{S}\mathcal{A}}}}$ follow  the same distribution and can be calculated by  \eqref{eq:PDFAU}. On the other hand, distance of link $\mathcal{A}$-$\mathcal{U}$  follows the different distribution  and can be computed by \eqref{eq:PDFSU1}. As  ${\gamma_{_{\mathcal{S}\mathcal{U}}}}$, ${\gamma_{_{\mathcal{S}\mathcal{A}}}}$ and ${\gamma_{_{\mathcal{A}\mathcal{U}}}}$ are the functions of random variables ${d_{_{\mathcal{S}\mathcal{U}}}}$, ${d_{_{\mathcal{S}\mathcal{A}}}}$, and ${d_{_{\mathcal{A}\mathcal{U}}}}$ respectively as defined by \eqref{eq: snr_E} and \eqref{eq: snr_J}, their PDFs can be obtained using random variable transformation.
  We can now obtain the  PDF of the ratio of ${\gamma_{_{\mathcal{S}\mathcal{U}}}}$ to ${\gamma_{_{\mathcal{S}\mathcal{A}}}}$ as follows:
 \begin{align} 
 \label{eq:PDFratio_E_Scentre}
%\begin{split}
{{f}}_{_{\bar{\gamma}^{^ {\mathcal{S}\mathcal{U}}}_{_{\mathcal{S}\mathcal{A}}}}}\left(z\right)&
=\begin{cases}
 { \frac{({\lambda_{\mathcal{E}}})^{-\frac{2}{\theta}}z^{\frac{2}{\theta}-1}} {\theta}},  &   \text{$0<z<{{\lambda_{\mathcal{E}}}}{},\ $}\\
  \frac{{({\lambda_{\mathcal{E}}})^{\frac{2}{\theta}}  z^{-\frac{2}{\theta}-1}}}{\theta}, & \text{$   {\lambda_{\mathcal{E}}}\ \leq  z <\infty.   $}\\
  %%(1-\alpha)\delta(z-\infty), & \text{$ z=\infty $} \\
\end{cases}
%\end{split}
\end{align}
Similarly,  the  PDF of ratio of  ${\gamma_{_{\mathcal{S}\mathcal{U}}}}$ and ${\gamma_{_{\mathcal{A}\mathcal{U}}}}$ can be  given as:
\begin{align}\label{eq:PDFratio_jamming_Scentre}
\begin{split}
&{f}_{_{\bar{\gamma}^{^ {\mathcal{S}\mathcal{U}}}_{_{\mathcal{A}\mathcal{U}}}}}\hspace{-1mm}\left(z\right)\\&= %\hspace{-0.5mm}
\begin{cases}
  \frac{ (\lambda_\mathcal{J})^{-\frac{2}{\theta}}z^{\frac{2}{\theta}-1} }{ \theta } -\frac{ 32(\lambda_\mathcal{J})^{-\frac{2}{\theta}}z^{-\frac{2}{\theta}-1} \mathrm{\beta}_1\left(\frac{(\lambda_\mathcal{J} )^{-\frac{2}{\theta}} }{4z^{-\frac{2}{\theta}}}\right)}{  \pi \theta },\\ \qquad \qquad \qquad \qquad \qquad \text{ $\forall \; 0<z<{2^\theta\lambda_\mathcal{J}}{},\ $}\\
 \frac {16(\lambda_\mathcal{J})^{\frac{2}{\theta}}z^{-\frac{2}{\theta}-1}R^2 }{   \theta D^2}- \frac{32({\lambda_\mathcal{J}})^{\frac{2}{\theta}}z^{-\frac{2}{\theta}-1}\left(  \mathrm{B}(\frac{1}{2},\frac{3}{2})-\mathrm{B}(\frac{5}{2},\frac{3}{2})\right)}{  \pi \theta }, 
 \\ \qquad \qquad \qquad \qquad \qquad 
 \text{$ \forall\; {2^\theta\lambda_\mathcal{J}}{}\ \leq  z <\infty.   $}
 \end{cases}
\end{split}
\end{align}
% where $\mathrm{\beta}_{4}= \left(  \mathrm{B}(\frac{1}{2},\frac{3}{2})-\mathrm{B}(\frac{5}{2},\frac{3}{2})\right). $ \\
% %
 \subsubsection{$\mathcal{U}$ at the origin}
    In this configuration, $\mathcal{S}$ and  $\mathcal{A}$ are located stochastically  within the deployment region  with $\mathcal{U}$ considered at centre. Here,
    distances  ${d_{_{\mathcal{S}\mathcal{U}}}}$ and ${d_{_{\mathcal{A}\mathcal{U}}}}$  follow  the same distribution and can be calculated with the help of  \eqref{eq:PDFAU}.  Distance distribution of link $\mathcal{S}$-$\mathcal{A}$  can be calculated with the help of \eqref{eq:PDFSU1}.     As explained above,  $f_{{\gamma_{_{\mathcal{S}\mathcal{U}}}}}$, $f_{{\gamma_{_{\mathcal{S}\mathcal{A}}}}}$ and $f_{{\gamma_{_{\mathcal{A}\mathcal{U}}}}}$  can be obtained using random variable transformation.
    Thus,  PDF of ratio of ${\gamma_{_{\mathcal{S}\mathcal{U}}}}$ and ${\gamma_{_{\mathcal{S}\mathcal{A}}}}$  is given by:
%%%%%%%%%%%%%%
\begin{align}\label{eq:PDFratio_eve_Ucentre}
\begin{split}
{f}_{_{\bar{\gamma}^{^ {\mathcal{S}\mathcal{U}}}_{_{\mathcal{S}\mathcal{A}}}}}\left(z\right)& =\begin{cases}
  \frac{ (\lambda_\mathcal{E})^{-\frac{2}{\theta}} z^{\frac{2}{\theta}-1} } {\theta } -\frac{ 32(\lambda_\mathcal{E})^{-\frac{2}{\theta}}z^{-\frac{2}{\theta}-1}\mathrm{\beta}_{1}\left(\frac{(\lambda_\mathcal{E} )^{-\frac{2}{\theta}} }{z^{-\frac{2}{\theta}}}\right) }{  \pi \theta },\\ \qquad \qquad \qquad \qquad \qquad \qquad \text{$ \forall \; 0<z<{{2}^{\theta}\lambda_\mathcal{E}} $},\\
 \frac {32(\lambda_\mathcal{E})^{\frac{2}{\theta}}z^{-\frac{2}{\theta}-1} }{  4 \theta } -\frac{32({\lambda_\mathcal{E}})^{\frac{2}{\theta}}z^{-\frac{2}{\theta}-1}}{  \pi \theta }\mathrm{\beta}_{4}, \\ \qquad \qquad \qquad \qquad \qquad \qquad \text{$ \forall \; {{2}^{\theta}\lambda_\mathcal{E}}\ \leq  z <\infty   $}.
\end{cases}
\end{split}
\end{align}
Similarly,  PDF of ratio of ${\gamma_{_{\mathcal{S}\mathcal{U}}}}$ and ${\gamma_{_{\mathcal{A}\mathcal{U}}}}$ can be  given as:
\begin{align}\label{eq:PDFratio_Jam_Ucentre}
\begin{split}
{{f}}_{_{\bar{\gamma}^{^ {\mathcal{S}\mathcal{U}}}_{_{\mathcal{A}\mathcal{U}}}}}\left(z\right) =\begin{cases}
  \frac{ ({\lambda_{\mathcal{J}}})^{-\frac{2}{\theta}} z^{\frac{2}{\theta}-1} }{\theta},  &   \text{$0<z<{{\lambda_{\mathcal{J}}}}{}\ $},\\
  \frac {{4(\lambda_{\mathcal{J}}})^{\frac{2}{\theta}}  z^{-\frac{2}{\theta}-1}}{\theta}, & \text{$   {{\lambda_{\mathcal{J}}}}\ \leq  z <\infty   $}.\\
  %%(1-\alpha)\delta(z-\infty), & \text{$ z=\infty $} \\
\end{cases}
\end{split}
\end{align}
To realise a practical system, QoS constraint of legitimate nodes such as maximum separation between  $\mathcal{S}$-$\mathcal{U}$ can be considered with these configurations as  explained in   Section 2 and  Section 3. Similarly, the concept of eavesdropping zone can be implemented to consider hardware constraints of attacker in eavesdropping mode.
As it is clear from \eqref{eq:PDFratio_E_Scentre} to  \eqref{eq:PDFratio_Jam_Ucentre}  that 
 ratio PDF ${f}_{_{\bar{\gamma}^{^ {\mathcal{S}\mathcal{U}}}_{_{\mathcal{S}\mathcal{A}}}}} $ and ${f}_{_{\bar{\gamma}^{^ {\mathcal{S}\mathcal{U}}}_{_{\mathcal{A}\mathcal{U}}}}} $ is available in closed-form for both configurations.
  Therefore, SOP  for the above-stated  configurations under eavesdropping and jamming mode respectively can be calculated in closed-form in high SNR regime as discussed in Section 4.  Moreover, system parameters can be designed  to get desired secrecy performance, as explained in the next section.
\subsection{Secrecy QoS Aware Perspective of Design Parameters}
{In the proposed work, the critical design parameters including $D$,
	$P_\mathcal{S}$, $R$, $r$ and $P_\mathcal{J}$ are leveraged to achieve the desired secrecy performance with SOP as the QoS parameter. 	 There is a minimum acceptable level of SOP used as performance metric, which needs to be maintained for the communication.}
The proposed framework can be utilised to attain the intended secrecy performance through the above-mentioned parameters classified into two categories: 
  \subsubsection{ Distance Dependent Parameters}
   These parameters include  the maximum separation between ${\mathcal{S}}$ and ${\mathcal{U}}$ denoted by $D$,  the eavesdropping range $r$  and the deployment range $R$.
    \begin{itemize}
    \item    {The distance parameter $D$:} It characterises the ${\mathcal{S}}$ to ${\mathcal{U}}$  communication range.     %It defines the maximum possible separation between ${\mathcal{S}}$ and ${\mathcal{U}}$, limiting the  communication range to $D$. 
    Larger communication range or say network coverage needs a greater value of $D$ which may result in degraded QoS arising out of increasing path-loss.
       This condition seeks a trade-off between the conflicting requirements of achieving higher communication  range while meeting the required QoS.
        For the given system model, the maximum value of $D$ is restricted to $D_{max}=2R$, the farthest possible distance between $\mathcal{S}$ and $\mathcal{U}$ in a circular coverage area.
    The key insights leveraging to  designing of $D$ for attaining desired secrecy performance are given as follows:
            
     \textit{a)} 
         {Increasing path-loss with larger $D$ results  in reduced capacity of $\mathcal{S}$-$\mathcal{U}$ link, which in turn yields degraded SC and higher SOP. Hence, $\hat{p}^i_{o}$ is a monotonically increasing function of  $D$. } 
           Therefore, an acceptable value of SOP ${{p}}_{o_{_{th}}}$ as per the desired QoS requirements can be obtained from \eqref{eq:outage_E_Th2_highsnr} by setting  $D=D_{_{{th}}}$  
   which is determined as    $D_{_{{th}}}=\left\lbrace D\mathrel{}\middle|\mathrel{} \left( \hat{p}^i_{o}={{p}}_{o_{_{th}}}\right) \wedge\left(0 \le D \le  D_{max} \right )\right \rbrace$. Here, $D_{_{{th}}}$ represents the maximum allowed $ D $ for maintaining  an acceptable SOP ${{p}}_{o_{_{th}}}$. However, it is not possible to obtain the explicit analytic solution for $D_{_{{th}}}$  because of   highly non-linear terms in \eqref{eq:outage_E_Th2_highsnr} and \eqref{eq:outage_J_Th3}. 
        Thus,  Golden-section (GS) based linear search [39] technique is used to find the solution within  $(0< D\leq D_{\max})$. {For $D<D_{_{{th}}}$, SOP will fall within the acceptable range.} Here, $\hat{p}^i_{o}$ represents lower bound on SOP ${p}^i_{o}$  under jamming and eavesdropping, with the condition of  positive SC, and   corresponding $D_{_{{th}}}$ thus provides upper bound on the maximum allowed distance between ${\mathcal{S}}- {\mathcal{U}} $ pair.
    
     \textit{b)}  Since, the maximum distance between $\mathcal{S}$  and $\mathcal{U}$  is upper-bounded by $D$, this results into a lower bound for SNR and respective capacity  $C=\log_2({1+\frac{\kappa_{_{_{\mathcal{S}\mathcal{U}}}}}{{D}^{\theta}}})$ which is guaranteed  for legitimate channel in the absence of $\mathcal{A}$. 
     Therefore, for a given  ${C_{st}}$,  there exists a unique threshold point $D_o = \left({\frac{\kappa_{_{_{\mathcal{S}\mathcal{U}}}}  }{2^{C_{st}}-1}}\right)^{\frac{1}{\theta}}$ such that for $ D \leq D_o$,  outage is attributed by the $\mathcal{A}$ only.
               % occurs only due to security failure because of  presence of $\mathcal{A}$,
    % For $ D> D_o$, both factors, $\mathcal{A} $ and increased propagation losses caused by larger $D$,  jointly contribute to the outage causing the SOP to  increase at a faster rate.
     For $ D> D_o$, SOP increases at a faster rate due to  the  joint contribution  $\mathcal{A} $ as well as increased propagation losses caused by larger $D$ towards  outage.
    %Thus,  $D_o$ reflects the boundary of the significant impact of $D$ on SOP.
    Thus, the appropriate setting of  $D$ affects the SOPs ${p}^\mathcal{E}_{o}$ and ${p}^\mathcal{J}_{o}$ and asymptotic SOPs $\hat{p}^\mathcal{E}_{o}$ and $\hat{p}^\mathcal{J}_{o}$ significantly.
                            $ D_o$ being a function of ${C_{st}}$ presents  the QoS-aware design threshold.\\
                    \textit{c):} Under eavesdropping, it is  observed further  that there exists another threshold    $D_{_{sat}}=\left({\frac{\lambda_{\mathcal{E}} r^\theta }{2^{C_{st}}}}\right)^{\frac{1}{\theta}}$ 
            for a given $C_{st}$ and $r$ such that  for  $D < \min(D_{_{sat}},D_o)$, $\hat{p}^{\mathcal{E}}_{o}$ does not vary with $r$.
              The reason behind this behaviour lies in the fact that for   $p_o^I$ gets vanish for $D \leq D_o$ and { $\alpha p_o^E$ becomes constant with respect to $r$  for $D <D_{_{sat}}$
              as depicted from the first case of \eqref{eq:outage_E_Th2_highsnr} after substituting the value of $\alpha$ from \eqref{eq:PDFSNRSA}}. %\triangleq\frac{r^2}{R^2}$}.
              Hence, by keeping $D < \min(D_{_{sat}},D_o)$,  the impact of $r$ can be neutralised on the security performance.
              %\textit{Hence, by keeping $D < \min(D_{_{sat}},D_o)$ legitimate nodes  can neutralize the impact of $r$ on their security performance.} 
   
     \item {The eavesdropping range $r$:} It  % the maximum distance upto which $\mathcal{A}$ can listen and decode the legitimate data,
         reflects the capability of $\mathcal{A}$ to eavesdrop the communication. The larger value of $r$  improves the likelihood of eavesdropping but also results in a reduction of eavesdropping link capacity caused by increased path-loss which creates unfavourable conditions for $\mathcal{A}$. It is also constrained by  overheads in hardware required  to decode the tapped signal.             It, therefore, becomes crucial for $\mathcal{A}$  to determine a suitable value of $r$ to maintain a trade-off between its contradictory nature.
              For desired settings of  $r$, significant inferences are drawn from analytic expression derived in $\eqref{eq:outage_E_Th2_highsnr}$ as follows:\\
        % We utilize analytic expression derived in $\eqref{eq:outage_E_Th2_highsnr}$  to draw some inferences  about desired settings of  $r$.\\
       \textit{a):} %We note that  $p^E_{o}>p^I_o$
       Since the presence of effective  $\mathcal{A}$ increases the outage probability, it results in $p^E_{o}>p^I_o$. We also note that $p^E_{o}$ is a decreasing function of $r$ because increased average path-loss of $\mathcal{S}$-$\mathcal{A}$ link  with increased $r$ leads to degraded $\mathcal{S}$-$\mathcal{A}$ link capacity and reduced outage probability or improved secrecy performance. Further, eavesdropping probability $\alpha$ is an increasing function of $r$. Using all these observations, it is revealed that first term of right-hand side is always positive while second one is always negative in 
           $\frac{\partial \hat{p}^\mathcal{E}_{o}}{\partial r}=\frac{2r(p^E_{o}-p^I_o)}{R ^2}+\frac{r^2}{R ^2}\frac{\partial p^E_{o}}{\partial r}$. 
       % we can observe that first term  is always positive while second term is always negative as explained above.
        This shows non-monotonic nature of  $\hat{p}^\mathcal{E}_{o}$ in $r$.  We have provided more insights with numerical investigations later in the result section. 
       % We have further investigated  this case numerically in result section. \\
        \textit {b):} The above observation also  reveal existence of a 
        threshold point for $r$ determined  as $r_{sat}\triangleq \left({\frac{2^{C_{st}} D^\theta }{\lambda_{\mathcal{E}}}}\right)^{\frac{1}{\theta}}$, above which, $\hat{p}^\mathcal{E}_{o}$ will exhibit a definite   trend with $r$ that is unfavourable for $\mathcal{A}$.
       % point $r_{sat}=\left({\frac{2^{C_{st}} D^\theta }{\lambda}}\right)^{\frac{1}{\theta}}$ such that
       For  $r>r_{sat}$, or  {equivalently $C_{st}< \log_2 \left(\frac{\lambda_{\mathcal{E}} r_{sat}^\theta}{D^{\theta}}\right)$,  $\alpha p^E_{o} $ becomes constant with respect to $r$ as clear from the first case of  $\eqref{eq:outage_E_Th2_highsnr}$. %
       Additionally, $p^I_o$ becomes zero with $r>r_{sat}$ and $D\leq D_o$ as  legitimate channel (in the absence of effective $\mathcal{A}$) can ensure the given secrecy rate  as discussed above,  $\hat{p}^\mathcal{E}_{o}$, therefore, becomes independent of $r$ under these conditions.
                  In other case, with $r>r_{sat}$ and $D> D_o$, $\hat{p}^\mathcal{E}_{o}$  exhibits a negative trend with $r$ due to decrease in $(1-\alpha)p^I_{o}$  with increasing $r$.  Here, $p^I_{o}$ is constant in $r$, but $(1-\alpha)$ is decreasing function of the same.}
                From this inference, it can be concluded that $\mathcal{A}$  can  reduce   its eavesdropping cost by restricting its $r$ to $r_{sat}$.
        % \textit{Thus,  $\mathcal{A}$ can use this inference to restrict its $r$ to $r_{sat}$ as there is no gain achieved in terms of posing secrecy threat to $\mathcal {U}$ by  increasing $r$ further.}
        \item { Deployment range $R$:}  It characterises of deployment region with radius $R$ in which  $\mathcal{S}$ and $\mathcal{U}$ may be  randomly positioned anywhere. A larger value of $R$ results in reduced capacity for legitimate link caused by the  value of average path-loss. It eventually results in degraded SC. Furthermore, in the presented system model, $R$ is  a radius  of a region where $\mathcal{A}$ positioned at the origin.  The larger value of $R $, therefore, will decrease the likelihood of eavesdropping for a given $r$ which   results in improved SC under eavesdropping mode. Under  jamming mode also, the average path-loss of $\mathcal{A}$-$\mathcal{U}$ channel increases with increasing $R $ resulting in enhanced secrecy.
        Thus, the relation of ${p}_{o}$ may not be monotonic with $R$. 
                 We also notice that $\hat{p}^{\mathcal{E}}_{o}$ is function of $\frac{r}{R}$ and $\frac{D}{R}$ for $ D \leq D_o$ whereas  $\hat{p}^{\mathcal{J}}_{o}$  is a function of  $\frac{D}{R}$. We,  therefore, conclude that $R$ acts as a normalising parameter for $ D\leq D_o$ under eavesdropping and for entire range of $D$ under jamming.  Consequently, for a given SOP requirement,   $R$ influences the maximum possible communication range covering the legitimate nodes. 
     \end{itemize}
    \subsubsection{ Power Controlling Parameters}
            This category include transmission power of $\mathcal{S}$ denoted by $P_{\mathcal{S}}$ and transmission power of              $\mathcal{J}$ denoted by $P_{\mathcal{J}}$. 
    
      %  This category include transmission power of $\mathcal{S}$ denoted by $P_{\mathcal{S}}$ and transmission power of              $\mathcal{J}$ denoted by $P_{\mathcal{J}}$ 
          %. Detailed analysis of their behaviour is given below:
       \begin{itemize}
                 \item    {$P_{\mathcal{S}} :$}  % It is well known that the capacity of a link grows with  the transmitted power of ${\mathcal{S}$ $ P_{\mathcal{S}}$.
   Under eavesdropping mode of $ {\mathcal{A}}$, $P_{\mathcal{S}}$ enhances the capacity of both $\mathcal{S}$-$\mathcal{U}$ and $\mathcal{S}$-$\mathcal{A}$ links by enhancing decoding rate. Being the difference of these two capacities, SC shows different behaviour with $P_{\mathcal{S}}$.  In the high SNR regime, SC saturates with $P_s$ and controlled by only channel gain ratio of $\mathcal{S}$-$\mathcal{U}$ to $\mathcal{S}$-$\mathcal{A}$ links as clear from \eqref{eq:C_e_hi_total}. Therefore, the impact of $ P_{\mathcal{S}}$ on secrecy performance   vanishes under eavesdropping mode in high SNR regime. On the other hand,  $P_{\mathcal{S}}$ enhances the capacity of only $\mathcal{S}$-$\mathcal{U}$ link under jamming mode while   $P_{\mathcal{J}}$ controls the  $\mathcal{A}$-$\mathcal{U}$ link.
   Consequently, for a given  $P_{\mathcal{J}}$, SC increases with $P_{\mathcal{S}}$ and  $\hat{p}^{\mathcal{J}}_{o}$ in \eqref{eq:outage_J_Th3} becomes monotonically decreasing function of $P_{\mathcal{S}}$. However, large  $P_{\mathcal{S}}$ causes high energy cost at ${\mathcal{S}}$, and it is also constrained by energy budget of  $\mathcal{S}$. Hence, to attain an acceptable SOP, the minimum required $P_{\mathcal{S}}$ defined as  $P_{{\mathcal{S}}_{th}}$ can be found as $ P_{{\mathcal{S}}_{th}}=\left\lbrace P_{\mathcal{S}}\mathrel{}\middle|\mathrel{} \left( \hat{p}^\mathcal{J}_{o}={{p}}_{o_{_{th}}}\right)\wedge\left(0\le P_{\mathcal{S}}\le P_{{\mathcal{S}}_{\max}}\right)\right \rbrace$ from  $\eqref{eq:outage_J_Th3}$  using GSS within lower and upper bounds $(0<P_{\mathcal{S}} \leq P_{{\mathcal{S}}_{\max}}) $ where $ P_{{\mathcal{S}}_{\max}} $ depends upon energy budget of $\mathcal{S}$.
  \item {$P_{\mathcal{J}} :$} The jamming power $P_{\mathcal{J}}$, represents the attacker capability of  $\mathcal{A}$ in  jamming mode. Though a larger value of $P_{\mathcal{J}}$ is favourable for $\mathcal{A}$, it is also constrained by energy budget of $\mathcal{A}$. Further, $P_{\mathcal{J}}$ increases the interference to reception  of legitimate signal. Therefore, $\hat{p}^{\mathcal{J}}_{o}$ in  \eqref{eq:outage_J_Th3}, is a monotonically increasing function with    ${P_{\mathcal{J}}}$ for a given  $P_{\mathcal{S}}$. Thus, to yield  a desired SOP ${{p}}_{o_{_{th}}}$ to legitimate nodes, the minimum required $P_{\mathcal{J}}$ defined as  $P_{{\mathcal{J}}_{th}}$ can be found as $P_{{\mathcal{J}}_{th}}=\left\lbrace P_{\mathcal{J}}\mathrel{}\middle|\mathrel{} \left( \hat{p}^\mathcal{J}_{o}={{p}}_{o_{_{th}}}\right)\wedge\left(0\le P_{\mathcal{J}}\le P_{{\mathcal{J}}_{\max}}\right)\right \rbrace$ from  $\eqref{eq:outage_J_Th3}$  using GSS within lower and upper bounds $(0<P_{\mathcal{J}} \leq P_{{\mathcal{J}}_{\max}}) $ where $ P_{{\mathcal{J}}_{\max}} $ depends upon energy budget of $\mathcal{A}$.
           More specifically 
            %  Therefore, it is advantageous for $\mathcal{A}$ to set minimum required $P_{\mathcal{J}}$ to pose desired secrecy outage to $\mathcal{U}$. 
              $\hat{p}^{\mathcal{J}}_{o}$ is increasing function of $\frac{P_{\mathcal{J}}}{P_{\mathcal{S}}}$.
              However, the rate of increment varies for different value of $\frac{P_{\mathcal{J}}}{P_{\mathcal{S}}}$.
              A subsequent investigation is done in the result section.

    \end{itemize}
   
\section{Numerical Results}
%In this section, we present numerical results for performance evaluation of proposed QoS-aware system model.
Here, we numerically investigate the secrecy  performance of  the proposed QoS-aware system model. {Unless explicitly mentioned, we have considered:  $P_{{\mathcal{S}}}=P_{{\mathcal{J}}}=100$ mW, $\sigma^2=-90$ dBm, $\theta=3$  %Accounting for other antenna and channel gains,%
$a_{_{\mathcal{S}\mathcal{U}}}=a_{_{\mathcal{S}\mathcal{A}}}=a_{_{\mathcal{A}\mathcal{U}}} = 10^{-5}$,  
 $C_{st}=1$ bps/Hz,  $r=50$ m, $R=100$ m referring IoT networks  \cite{IoTsecurity, White}. We consider that $\mathcal{S}$ and $\mathcal{U}$ are deployed stochastically uniformly within
a circle of radius $R$ with attacker $\mathcal{A}$  at its centre. The proposed configuration represents  a practical  ad-hoc network scenario  where   nodes can be placed randomly in the deployment region \cite{D2D}, and an attacker invades the network with a malicious intention of eavesdropping or jamming. The centre location of $\mathcal{A}$ will offer leverage  to it the maximal possible attacking capability. In other words, it depicts the worst-case scenario in terms of secrecy performance. To deal with all possible deployment cases, we have considered that legitimate nodes may lie inside the eavesdropping zone of radius $r$  or outside of  it.  The maximum distance between $\mathcal{S}$ and $\mathcal{U}$ is restricted to $D$.}
We have  considered  $D=r$ under eavesdropping and  $D=R$ under jamming in validation with the purpose of fair comparison of other parameters by keeping maximum distance for legitimate and attacker links same. For the given system parameters, the value of distance  threshold parameter $D_o$ has been determined to $100$ m  following discussion in Section 5. 
 \subsection{Validation of analysis}
The analytical formulations  carried out in Section 3   lay the foundation for secrecy analysis performed under Section 4.
%\vspace{-2mm}
{ Figs. \ref{fig:evepdf} and \ref{fig:Jampdf}  validate these analytical formulations through extensive Monte-Carlo simulations. Analytical results are represented by different line styles for different parameters while marker `o' is used to represent corresponding simulation results. For generating simulation results,  we have used $10^6$ random realisations of distances  ${d_{_{\mathcal{S}\mathcal{U}}}}$, ${d_{_{\mathcal{S}\mathcal{A}}}}$ and ${d_{_{\mathcal{A}\mathcal{U}}}}$ by creating random locations of $\mathcal{S}$ and $\mathcal{U}$.}\begin{figure}[!h]
%\centering\includegraphics[width=2.5in]{Figures/eve_pdf3_newNotation.eps}
\centering\includegraphics[width=2.8in]{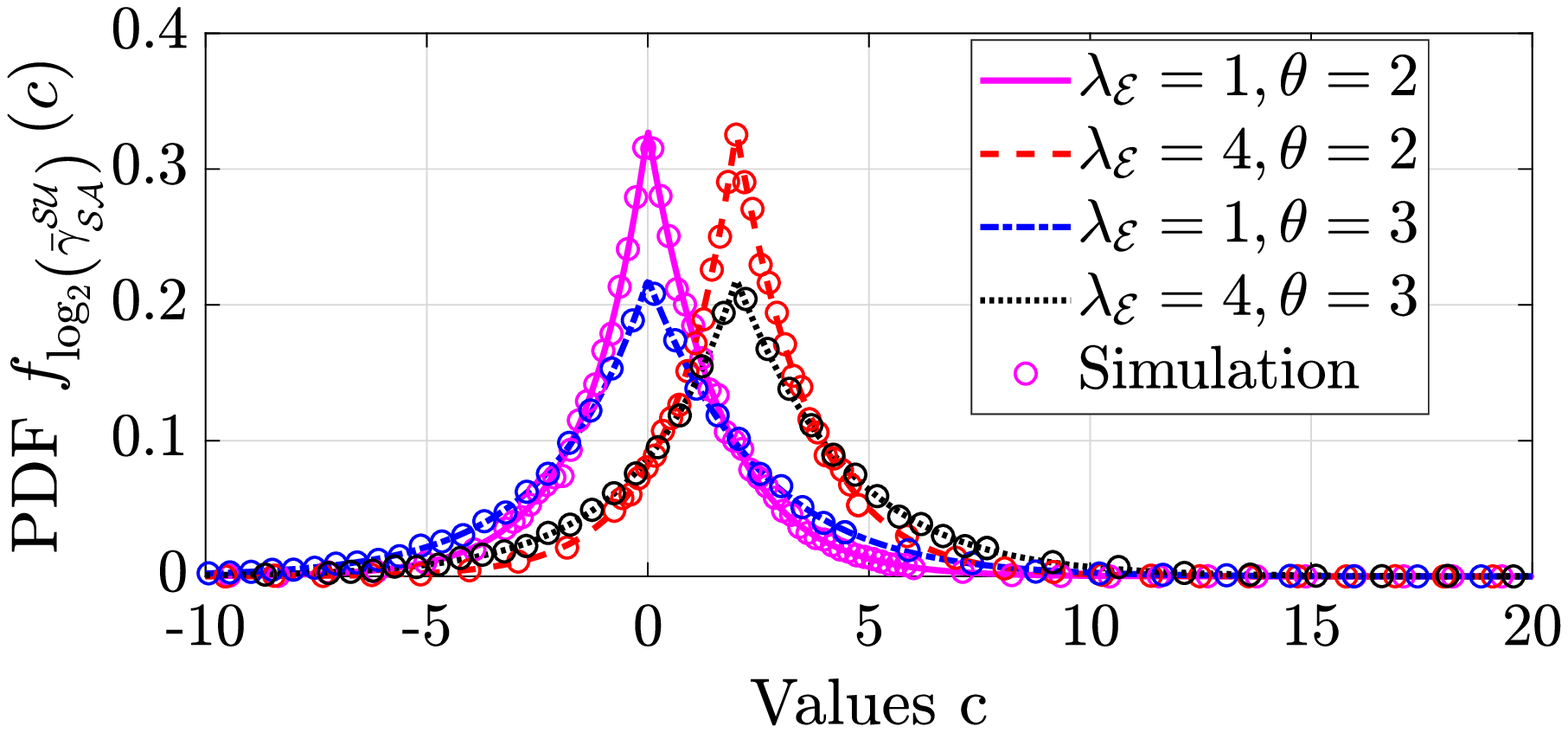}
%%%call your figure name in the place "figurename.eps"
\caption{ \footnotesize \textit{Validation  for  PDF of  ${\log_2\left({\bar{\gamma}^{^ {\mathcal{S}\mathcal{U}}}_{_{\mathcal{S}\mathcal{A}}}}\right)}$  with  $D=r$ under eavesdropping.}}
\label{fig:evepdf} 
\vspace{-2mm}
\end{figure}
%%%%%%%%%%%%%Fig Jamming PDF
\begin{figure}[!h]
%\centering\includegraphics[width=2.5in]{Figures/jam_pdf3_newnotation.eps}
\centering\includegraphics[width=2.8in]{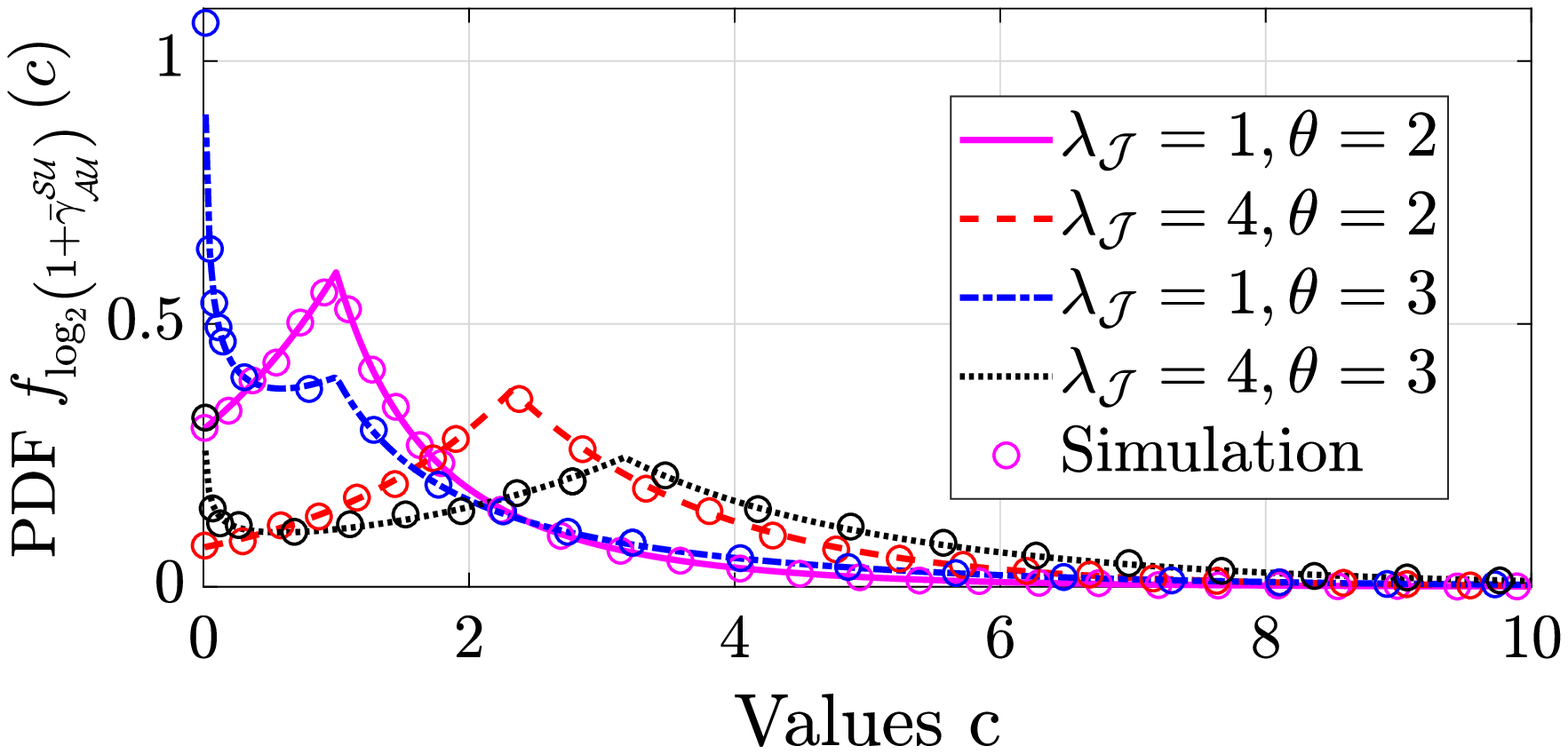}
%%%call your figure name in the place "figurename.eps"
\caption{\footnotesize \textit{Validation  for  PDF of  ${\log_2\left(1+{\bar{\gamma}^{^ {\mathcal{S}\mathcal{U}}}_{_{\mathcal{A}\mathcal{U}}}}\right)}$   with $D=R$   under jamming.}}
\label{fig:Jampdf}
\vspace{-2mm}
\end{figure}
Firstly, the expression \eqref{eq:PDFSC1} for PDF  ${{f}}_{\log_2\left({\bar{\gamma}^{^ {\mathcal{S}\mathcal{U}}}_{_{\mathcal{S}\mathcal{A}}}}\right)}\left(c\right) $ is  validated in Fig. {\ref{fig:evepdf}}. It also represents the PDF of  SC $ {\hat{C}}^{\mathcal{E}}_{_{1,2}}$  for the cases 1 and 2 under eavesdropping mode under high SNR regime     as observed from  \eqref{eq:C_e_hi_total} provided  ${\log_2\left({\bar{\gamma}^{^ {\mathcal{S}\mathcal{U}}}_{_{\mathcal{S}\mathcal{A}}}}\right)\geq 0}$. Hence, it is represented as   ${{f}}_{\hat{C}^{\mathcal{E}}_{_{1,2}}}\left(c\right) $ for further discussion. 
We have used $10^6 $ random realisations of distances $d_{_{\mathcal{S}\mathcal{U}}}$ and $d_{_{\mathcal{S}\mathcal{A}}} $ to generate 
simulation results.  Analytical and simulation results are found to be  in close agreement  with  root mean square error (RMSE) $ < 10^{-4}$.
 Similarly,  Fig. {\ref{fig:Jampdf}}, validates the   analytical expression for PDF  ${f}_{\log_2\left(1+{\bar{\gamma}^{^ {\mathcal{S}\mathcal{U}}}_{_{\mathcal{A}\mathcal{U}}}}\right)}\left(c\right) $ given by \eqref{eq:PDFratio_jamming_SC} which also represents  PDF of SC ${\hat{C}^{\mathcal{J}}}$ under jamming mode in high SNR regime, as explained in \eqref{eq:C_J_asym2}. Therefore, equivalent representation of PDF can be expressed as ${{f}}_{\hat{C}^\mathcal{J}}\left(c\right) $. A close match is observed between analytical   and simulation results 
%  generated using $10^6$ random realisations of distances $d_{_{\mathcal{S}\mathcal{U}}}$ and $ d_{_{\mathcal{A}\mathcal{U}}}$ 
 with  RMSE  $ <  10^{-4}$. 
  We have validated the PDFs ${{f}}_{\hat{C}^{\mathcal{E}}_{_{1,2}}}\left(c\right) $ and $f_{{\hat{C}}^{\mathcal{J}}} $ are  for various values of  ratio $\lambda_{\mathcal{E}}$, $\lambda_{\mathcal{J}}$  respectively and analysed the impact of path-loss exponent $\theta$ also.  Here, $\lambda_{\mathcal{E}}$ reflects the  ratio of channel parameters $a_{_{\mathcal{S}\mathcal{U}}}$ to $a_{_{\mathcal{S}\mathcal{A}}}$ and  $\lambda_{\mathcal{J}}$ represents ratio of channel parameters $a_{_{\mathcal{S}\mathcal{U}}}$ to $a_{_{\mathcal{A}\mathcal{U}}}$ for $P_{{\mathcal{S}}}=P_{{\mathcal{J}}}$.
The effect of relative channel conditions of legitimate and both attacker links are revealed in 
Figs. {\ref{fig:evepdf}} and {\ref{fig:Jampdf}} %reveal the impact of relative channel conditions of legitimate and attacker links also. 
Higher channel gain of legitimate link compared to that for legitimate link is found to yield larger realisations of  ${\hat{C}}^{\mathcal{E}}_{_{1,2}}$ and ${\hat{C}}^{\mathcal{J}}$  in eavesdropping and jamming respectively.
% Poorer channel gain  of attacker link compared to that for legitimate link is observed to yield larger realisations of  ${\hat{C}}^{\mathcal{E}}_{_{1,2}}$ and ${\hat{C}}^{\mathcal{J}}$  in eavesdropping and jamming respectively.
 Larger spread in  both distributions is also observed for a larger value of path-loss exponent $\theta $.
Since asymptotic SOP $\hat{p}^\mathcal{E}_{o}$ and $\hat{p}^\mathcal{J}_{o}$ in $\eqref{eq:outage_E_Th2_highsnr}$ and  $\eqref{eq:outage_J_Th3}$ respectively are  functions of asymptotic SC, SOP analysis carried out in Section 4 is also verified by 
analytical validation of SC.
%%%%%%%%%%%%%% SNR of eavesdropping
\begin{figure}[!h]
%\centering\includegraphics[width=2.5in]{Figures/snr_eav_journal.eps}
%\centering\includegraphics[width=2.5in]{Figures/hihhSNReavesdropping.eps}
\centering\includegraphics[width=2.8in]{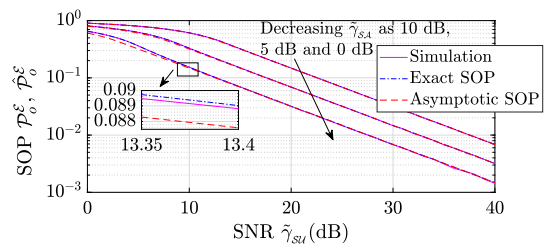}
%%%call your figure name in the place "figurename.eps"
\caption{\footnotesize{\textit{ {SOP validation for $R\!=\! r\!=\!D$  under eavesdropping.}}}}
%\caption{\textit{\small{SOP validation for $R\!=\! r\!=\!D$  under eavesdropping.}}}
\label{fig:exact_eavesdropping} 
\vspace{-2mm}
\end{figure}

%%%%%%%%%%%%%%%%%%% SNR of jamming
\begin{figure}[!h]
\centering
\includegraphics[width=2.8in]{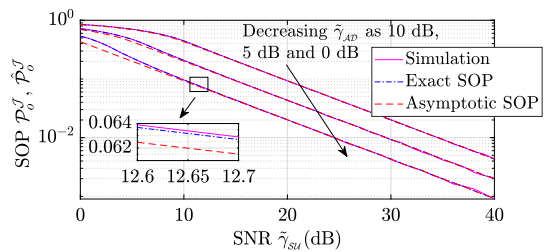}
%%%call your figure name in the place "figurename.eps"
\caption{SOP validation  $R=D$  under jamming.}
\label{fig:exact_jamming} 
\vspace{-2mm}
\end{figure}
{Figs. \ref{fig:exact_eavesdropping} and \ref{fig:exact_jamming} verify the accuracy of closed-form solutions given by \eqref{eq:outage_E_Th2_highsnr} and \eqref{eq:outage_J_Th3} for SOPs derived through asymptotic analysis. They are compared    with  the exact SOPs given by  \eqref{eq:outage_E_Th1} and \eqref{eq:outage_J_exact} in Section 4 under eavesdropping and jamming modes respectively. Analytical results for the exact SOP are also compared with Monte-Carlo simulation results and found in close agreement.} Fig. \ref{fig:exact_eavesdropping} shows variation of SOP ${p}^\mathcal{E}_{o}$ and asymptotic SOP $\hat{p}^\mathcal{E}_{o}$  against worst case SNRs  as denoted by ${\tilde{\gamma}_{_{\mathcal{S}\mathcal{U}}} }$ and  $\tilde{\gamma}_{_{\mathcal{S}\mathcal{A}}}$. Similarly, variation of SOPs  ${p}^\mathcal{J}_{o}$ and asymptotic SOP $\hat{p}^\mathcal{J}_{o}$  are plotted in Fig. \ref{fig:exact_jamming}  against worst case SNRs  $ {\tilde{\gamma}_{_{\mathcal{S}\mathcal{U}}} }$ and   $\tilde{\gamma}_{_{\mathcal{A}\mathcal{U}}}$. Hereby,   $\tilde{\gamma}_{_{\mathcal{S}\mathcal{U}}}$ = $\frac{\kappa_{_{\mathcal{S}\mathcal{U}}}} {D^{\theta}}$,  $\tilde{\gamma}_{_{\mathcal{S}\mathcal{A}}} $ = $\frac{\kappa_{_{\mathcal{S}\mathcal{A}}}} {r^{\theta}} $, and $\tilde{\gamma}_{_{\mathcal{A}\mathcal{U}}} $ = $\frac{\kappa_{_{\mathcal{A}\mathcal{U}}}} {R^{\theta}} $. 
 For $\tilde{\gamma}_{_{\mathcal{S}\mathcal{A}}}=\tilde{\gamma}_{_{\mathcal{A}\mathcal{U}}}=0$ dB, we observe that results are closely matching for  $\tilde{\gamma}_{_{\mathcal{S}\mathcal{D}}}>7$ dB under eavesdropping and $\tilde{\gamma}_{_{\mathcal{S}\mathcal{D}}}>9$ dB under jamming. It is also revealed that under the condition $\tilde{\gamma}_{_{\mathcal{S}\mathcal{A}}}=\tilde{\gamma}_{_{\mathcal{A}\mathcal{U}}}=5$ dB,  results are  in good agreement  under eavesdropping while very low deviations are observed for   $\tilde{\gamma}_{_{\mathcal{S}\mathcal{D}}}<5$ dB under jamming. Hence,  closed-form results derived for asymptotic SOP  holds good for evaluating SOP for wide  range of practical values of SNR for legitimate as well as both attacker channels. Henceforth,  we   refer the asymptotic SOP as simply SOP.

\subsection{Insight about design parameters}
Figs. \ref{fig:r_SOP} to \ref{fig:DVsRjam1} deliver the valuable insights about deigning of system  parameters for achieving intended secrecy performance in  accordance with  the  discussion  carried out  in  Section  5.   The impact  of  maximum  distance $D$ between legitimate nodes  on  SOP is analysed in Fig \ref{fig:r_SOP} with varying eavesdropping range $r$ of $\mathcal{A}$. {Since $r$ and $D $  reflect the restriction on random  inter-node distances respectively  ${d_{_{\mathcal{S}\mathcal{A}}}}$ for effective eavesdropping and  ${d_{_{\mathcal{S}\mathcal{U}}}}$ for maintaining QoS under eavesdropping as well as jamming, we may take extreme values of $r=r_{max}=R$ and $D=D_{max}=2R$   to represent a traditional scenario when no restriction is put on corresponding distances in terms of these variables.}
% We observe that  SOP $ \hat{p}^\mathcal{E}_{o}$ records
  A significant rise in SOP $ \hat{p}^\mathcal{E}_{o}$ is recorded for $D > D_o$.  
  %The results also reveal that $D$ has relatively more adverse impact on SOP $ \hat{p}^\mathcal{E}_{o}$ than $r$.
  The results also reveal that $D$ influence SOP $ \hat{p}^\mathcal{E}_{o}$ more adversely as compared to $r$.
%   has relatively more adverse impact on SOP $ \hat{p}^\mathcal{E}_{o}$ than $r$.
 We also note that, satisfying the condition $r$ > $r_{sat}$ results in constant SOP for $D \leq D_o$ where $r_{sat}$ is defined in Section 5.
%   For  $D \leq D_o$, SOP becomes steady after a particular value of $r$ defined as $r_{sat}$ in Section V. 
  For a typical value of $D$= 60 m, SOP saturates to 0.11 for $r> r_{sat}=75$ m and  becomes independent of a further increase in eavesdropping range.  Fig. \ref{fig:r_SOP} also exhibits the uni-modal nature of $\hat{p}^\mathcal{E}_{o}$ in $r$.   It shows that there exists a value of $ r$    that gives optimal performance. 
   %The attacker  $\mathcal{A}$ can set  this optimal value of $r$   for posing maximum adverse effect on the secrecy of user communication.
    {This optimal value of $r$  can be utilised by $\mathcal{A}$ for posing maximum adverse impact on the secrecy performance of legitimate communication in an efficient manner as compared to its extreme value.}
    \begin{figure}[!h]
\centering
\includegraphics[width=2.5in]{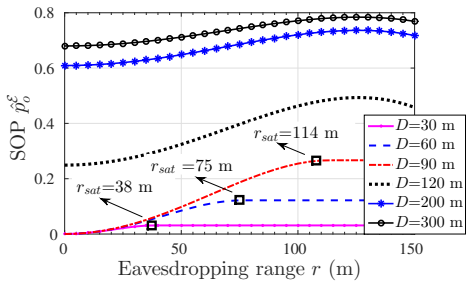}
%{{\includegraphics[width=2.5in]{Figures/EavD_RFig7.eps}} }%\vspace{-5mm}
%%%call your figure name in the place "figurename.eps"
\caption{\footnotesize \textit{Insights about $r$ with $R\!=\!150$ m, $D_o\!=\!100$ m under eavesdropping.}}
\label{fig:r_SOP} 
\vspace{-8mm}
\end{figure}
%%%%%%%%%%D
\begin{figure}[!h]
\vspace{-2mm}
\centering
\includegraphics[width=2.7in]{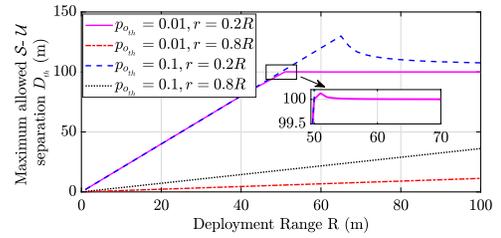}
%%%call your figure name in the place "figurename.eps"
\caption{\footnotesize \textit{Maximum allowed $D$ under  eavesdropping.}}
\label{fig:D_R} 
%\vspace{-2mm}
\end{figure}
\\
\begin{remark}\label{rem:constSOP}
 \textit{ {The SOP $ \hat{p}^\mathcal{E}_{o}$  under  eavesdropping can be retained at a remarkably lower value by meeting the conditions $D \leq D_o$. Furthermore,  from attacker point of view, eavesdropping range of $\mathcal{A}$ can be restricted to $ r_{sat}$ instead of its extreme value
 %$\mathcal{A}$ can restrict its eavesdropping range  to $ r_{sat}$, 
 for $ D< D_o$ as no further gain is attained by  $\mathcal{A}$ towards  secrecy outage.} Equivalently, $\mathcal{S}$-$\mathcal{U}$ pair can restrict the maximum separation between them to $D_{_{sat}}$ so as  to neutralise the impact of eavesdropping range on   $\hat{p}^\mathcal{E}_{o}$  beyond a given value of $r$ provided  $D_{_{sat}} \leq D_{o}$.    It may, therefore,  be recognised that for  technologies used for short-range communication like D2D and IoT, the effect of $r$ gets saturated beyond a certain value  while an important role is played by $r$ in secrecy performance for technologies used for long range communications. }
  \end{remark}
 Fig. \ref{fig:D_R} provides the insights about $D_{_{{th}}}$, the maximum value of allowed $D$  to achieve an acceptable value of SOP ${{p}}_{o_{_{th}}}$  for various values of ${{p}}_{o_{_{th}}}$ and  $r$ with varying $R$. % It is clear from the figure that  maximum  $D$ allowed to achieve   acceptable SOP  increases with increment in acceptable SOP hence legitimate nodes can enjoy higher communication range.  
 %Value of ${{p}}_{o_{_{th}}}$ are set as 0.01 and 0.1.
 It is observed that $D_{_{{th}}}$ exhibits increasing nature  with ${{p}}_{o_{_{th}}}$ and  decreasing trend  with $r$. 
 We also find that the value of   $D_{_{{th}}}$ linearly increases  with increased $ R$ for $D_{_{{th}}}<D_o$ or $D_{_{{th}}}=D_{max}$.  The former condition happens because of the fact that for $D_{_{{th}}}<D_o$, ${{p}}_{o_{_{th}}}$ is a function of $\frac{D_{_{{th}}}}{R}$. The latter condition is based on the fact that acceptable value of  SOP ${{p}}_{o_{_{th}}}$ is high enough that  maximum communication range can be enjoyed by legitimate nodes for $D_{_{{th}}}=D_{max}$ which is found to be a linear function of $R$.
 Under the former condition, when $D_{_{{th}}}$ exceeds $D_o$ with increasing $R$, the behaviour of $D_{_{{th}}}$ is illustrated for practical range of $\hat{{p}}_{o_{_{th}}}$  in Fig. \ref{fig:D_R}. It is noticed that, $D_{_{{th}}}$  fastly decreases first to meet acceptable SOP ${{p}}_{o_{_{th}}}$ and then gradually decreases with $R$.  This illustration, thus, can be used as a reference for the design of secrecy performance aware  $D_{_{{th}}}$.
   \begin{remark}\label{rem:deployementRangeE}
 \textit{  Under eavesdropping, the large deployment range supports in the greater   maximum allowed $\mathcal{S}- \mathcal{U}$ separation $D_{_{{th}}}$ for a given acceptable SOP ${{p}}_{o_{_{th}}}$ provided   $D_{_{{th}}}<D_o$ or $D_{_{{th}}}=D_{max}$.  For $ D_o\leq D_{_{{th}}}<D_{max}$, there is no advantage in increasing  $R$ further for the purpose of enhancing secrecy QoS aware  $D_{_{{th}}}$ for a practical range of acceptable SOP.
%  it is not beneficial to further increase $R$ in order to enhance secrecy  aware  $D_{_{{th}}}$ for practical range of acceptable SOP.
   }
   \vspace{-2mm}
  \end{remark}
  \begin{figure}[!h]
  %\vspace{-2mm}
\centering
{{\includegraphics[width=2.7in]{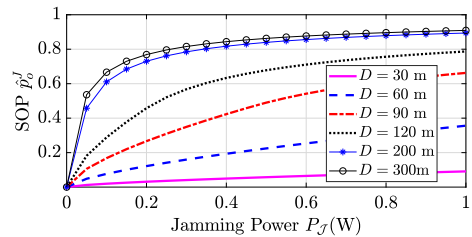}} }
%{{\includegraphics[width=2.5in]{Figures/EavD_RFig7.eps}} }%\vspace{-5mm}
%%%call your figure name in the place "figurename.eps"
\caption{\footnotesize \textit{SOP under jamming  for  $R\!=\!150$ m,  $D_o=100$ m.}}
\label{fig:PjVsSOP} 
\vspace{-2mm}
\end{figure}

   \begin{figure}[!h]
   \vspace{-2mm}
\centering
\includegraphics[width=2.7in]{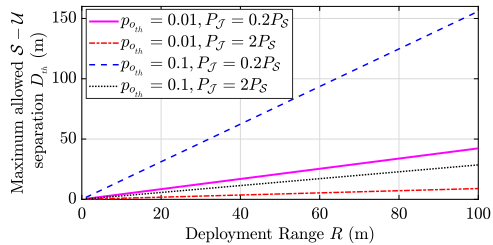}
%%%call your figure name in the place "figurename.eps"
\caption{\footnotesize \textit{Maximum   allowed $D$   under jamming.}}
\label{fig:DVsRjam1} 
%\vspace{-4mm}
\end{figure}
  
   In Fig. \ref{fig:PjVsSOP} we plot  variation of SOP  $\hat{p}^\mathcal{J}_{o}$ with $\mathcal{A}$'s jamming power $ P_{\mathcal{J}}$  and maximum distance between legitimate nodes $D$. 
  Results show that for $D \leq D_o$, $\hat{p}^\mathcal{J}_{o}$   increases almost linearly with increased $ P_{\mathcal{J}}$. For $D>D_o$,  there is a sharp increase in $\hat{p}^\mathcal{J}_{o}$. After a point,  $\hat{p}^\mathcal{J}_{o}$ increases gradually and gets saturated with $P_{\mathcal{J}}$. 
  For a particular value of  $D> 200$ m and $ P_{\mathcal{J}}>  0.5$ W, $\hat{p}^\mathcal{J}_{o}$ does not increase significantly and 
  almost gets saturated at $ P_{\mathcal{J}}>  1.5$ W which implies that $\mathcal{A}$ gets diminished returns by increasing $P_{\mathcal{J}}$ beyond this. Moreover, the impact of $D$  is found more detrimental on $ \hat{p}^\mathcal{J}_{o}$ relative to $P_{\mathcal{J}}.$
\begin{remark}\label{rem:Jam}
 \textit{   { The setting of $D \leq D_o$ enables us  to retain a lower level of SOP $ \hat{p}^\mathcal{J}_{o}$ in the presence of $\mathcal{A}$ in jamming mode.} % For  $D \leq D_o$, $\mathcal{A}$ can degrade the secrecy performance  by increasing  $P_{\mathcal{J}}$ consistently. 
 Though, degradation is significant for $D > D_o$,  $\mathcal{A}$ gets diminishing returns  by increasing $P_{\mathcal{J}}$ beyond a certain point. Therefore, in short-range  communication technologies like D2D, jamming power plays a more important role in degrading secrecy performance. Whereas in long-range communication technologies, the role of jamming power diminishes after a certain point.}
\end{remark}
 {These observations reflect that proposed model aids to analyse secrecy performance in a better way  as compared to the traditional scenario when  real-life constraints in terms of distance thresholds are not considered.}
Then the maximum allowed $\mathcal{S}- \mathcal{U}$ separation ${D_{_{th}}}$,  for the acceptable  SOP ${{p}}_{o_{_{th}}}= 0.01$ and $0.1$ are plotted  with different values of $P_{\mathcal{J}}$ in  Fig. \ref{fig:DVsRjam1} for jamming mode.
%   Fig. \ref{fig:DVsRjam1} analyses  ${D_{_{th}}}$, the maximum allowable $ {D} $ for an acceptable  SOP ${{p}}_{o_{_{th}}}$ under jamming  with different values of $P_{\mathcal{J}}$.
 It is observed that  ${D_{_{th}}}$  increases linearly with $ R$ for a given acceptable  SOP ${{p}}_{o_{_{th}}}$ because 
 ${{p}}_{o_{_{th}}}$ is a function of $\frac{D_{_{{th}}}}{R}$. This ${D_{_{th}}}$ is found as an  increasing function of ${{p}}_{o_{_{th}}}$  and a decreasing function of  $P_{\mathcal{J}}$. 
 The impact of ${{p}}_{o_{_{th}}}$ on $D_{_{{th}}}$ gets reduced with increased $P_{\mathcal{J}}$.
 \begin{remark}\label{rem:deployementRangeJ}
 \textit{  The larger deployment range results in  large value of maximum allowed  $\mathcal{S}- \mathcal{U}$ separation $D_{_{{th}}}$ for an acceptable SOP ${{p}}_{o_{_{th}}}$  under jamming. }  
  \end{remark}
 
%%%%%%%%%%%
{To analyse the impact of randomness in channel fading gain on  designing of $D_{th}$, we plot its PDF  $f_{D_{th}}(x)$ against its values $x$  under eavesdropping and jamming    in respectively Figs. 10 (a) and (b).
%\ref{fig:10a} and \ref{fig:10b}.
%for %$R=100$. 
We consider the random channel power gains due to Rayleigh fading with the same mean value as in Figs. \ref{fig:D_R} and \ref{fig:DVsRjam1}. Channel fading gains $\left|h_{_{ij}}\right|^2, \; \forall\; i,j=\left\lbrace\mathcal{S},\mathcal{U},\mathcal{A}\right\rbrace$ 
%for $\mathcal{S}$-$\mathcal{U}$, $\mathcal{S}$-$\mathcal{A}$ and $\mathcal{A}$-$\mathcal{U}$ links
follow exponential distribution. We also find that, $D_{th}$ increases with an increased ratio of fading gains  $\frac{|{h}_{_{\mathcal{S}\mathcal{U}}}|^2}{|{h}_{_{\mathcal{S}\mathcal{A}}}|^2}$ for eavesdropping and  $\frac{|{h}_{_{\mathcal{S}\mathcal{U}}}|^2}{|{h}_{_{\mathcal{A}\mathcal{U}}}|^2}$ for jamming. 
%$\mathcal{S}$-$\mathcal{U}$ link to $\mathcal{S}$-$\mathcal{A}$ link. 
This can be explained by the fact that SOP decreases with an increased ratio of channel gains $\frac{|\mathrm{g}_{_{\mathcal{S}\mathcal{U}}}|^2}{|\mathrm{g}_{_{\mathcal{S}\mathcal{A}}}|^2}$ and $\frac{|\mathrm{g}_{_{\mathcal{S}\mathcal{U}}}|^2}{|\mathrm{g}_{_{\mathcal{A}\mathcal{U}}}|^2}$ for eavesdropping and jamming respectively.  Channel gains consist of a ratio of fading gain to the path loss, i.e., $\left|\mathrm{g}_{_{ij}}\right|^2= \frac{\left|h_{_{ij}}\right|^2}{\left(d_{_{ij}}\right)^\theta},\;\forall \; i,j=\left\lbrace\mathcal{S},\mathcal{U},\mathcal{A}\right\rbrace$ where ${d_{_{\mathcal{S}\mathcal{U}}}} \leq D,  {d_{_{\mathcal{S}\mathcal{A}}}} \leq r,   {d_{_{\mathcal{A}\mathcal{U}}}} \leq R$.  
Therefore,  for a given acceptable SOP ${{p}}_{o_{_{th}}}$, an increase in fading gain ratio leads to an increase in acceptable value of path loss in $\mathcal{S}$-$\mathcal{U}$ link due to ${d_{_{\mathcal{S}\mathcal{U}}}}$ for the given parameters. This, in turn,  increases the maximum allowed $\mathcal{S}$-$\mathcal{U}$ separation $D_{th}$. These findings lead to the corresponding pdf  of $D_{th}$  as observed in Figs. 10 (a) and 10 (b).% \ref{fig:10a} and \ref{fig:10b}. 	
Increasing values of ${{p}}_{o_{_{th}}}$ shift the curve towards right in the sense that the larger realisations of $D_{th}$ become more likely and vice versa.}

	\begin{figure}[!t]
  \centering
  \vspace{-1mm}
  \subfigure[\textit{\footnotesize Eavesdropping with $r$ = $0.8R$}]
   %{{\includegraphics[width=1.8in]{Figures/Fig1R12.eps} }}\hspace{-5mm}
  % {{\includegraphics[width=2.6in]{Figures/EavPDFDnormalisedR.eps} }}
   { \centering\includegraphics[width=1.6in]{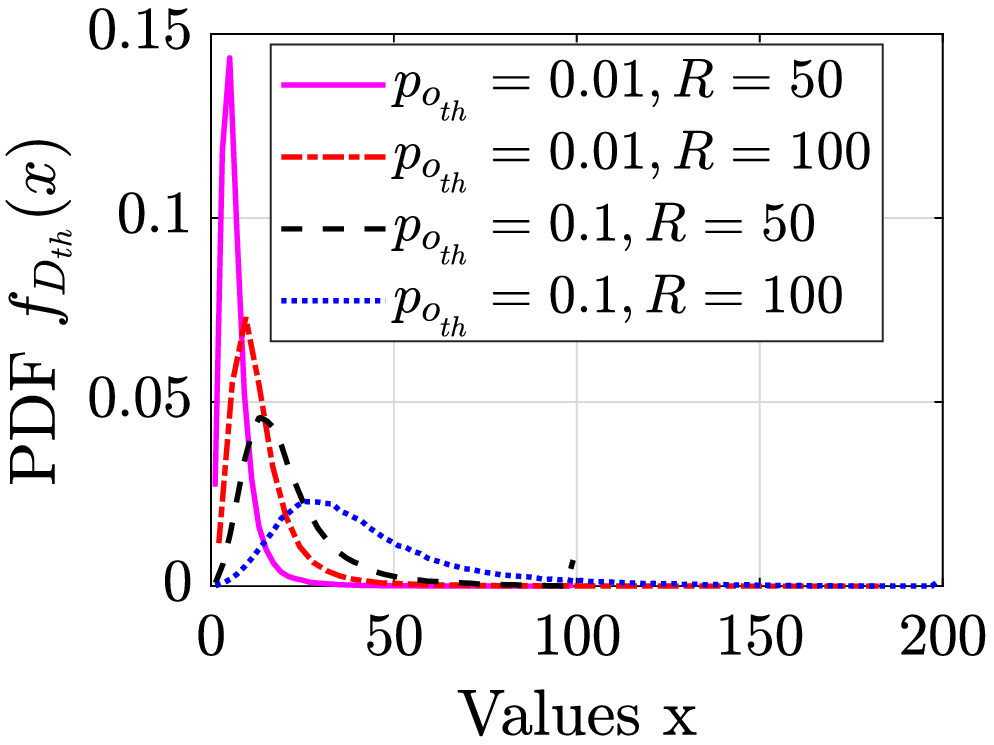}}
    \label{fig:10a}
   %\hspace{5mm}
 \vspace{-1mm} 
 \subfigure[\textit{ Jamming with $P_{\mathcal{J}}= 2 P_{\mathcal{S}}$}] 
  %   \hspace{-2mm}
 { \centering\includegraphics[width=1.6in]{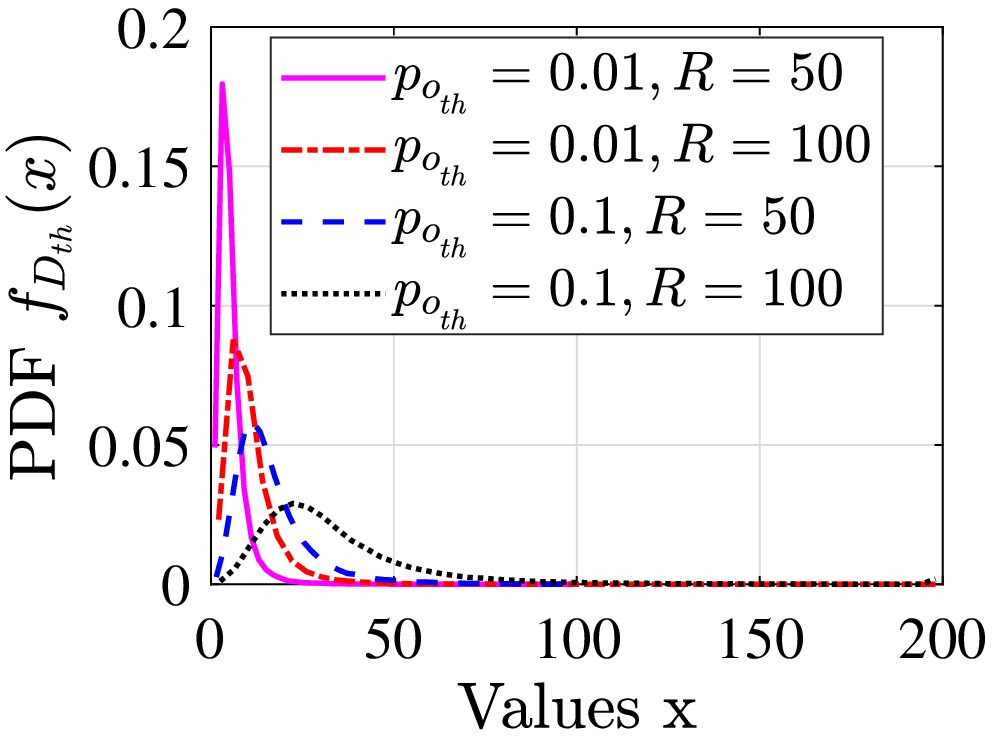}}
 \label{fig:10b}
  %\vspace{-4mm}
 % \centering
       \caption{ \footnotesize \textit{ Impact of channel fading randomness on  $D_{th}$.} }%
    %\label{fig:optimization} 
    %\vspace{-9mm}
\end{figure}
\subsection{Relative Severity of Eavesdropping and Jamming}
As discussed in Section 5, for a given secrecy rate threshold $C_{st}$ and channel parameters, there exists a corresponding threshold   for distance threshold  $D_o$ such that for  $C_{st}\leq \log_2({1+\frac{\kappa_{_{_{\mathcal{S}\mathcal{U}}}}}{{D}^{\theta}}})$ or equivalently $D \leq D_o$, outage is caused by   $\mathcal{A}$ only. For this scenario, under eavesdropping   $p_o^\mathcal{E}$ initially increases  with  increased $r$ and 
saturates after reaching $\alpha=r^2/R^2$. For $C_{st}>\log_2({1+\frac{\kappa_{_{_{\mathcal{S}\mathcal{U}}}}}{{D}^{\theta}}})$ or equivalently $D>D_o$,  there is a sharp increase in 
SOP %increases at a faster rate
caused by randomness in propagation losses along with security breach due to presence of $\mathcal{A}$ because there is no guarantee to  support the high secrecy rate  $C_{st}$ by legitimate link.
SOPs $\hat{p}^\mathcal{E}_{o}$ and $\hat{p}^\mathcal{J}_{o}$  increases with $r$ and $P_{\mathcal{J}}$ in  eavesdropping and jamming respectively. 
%We observe that for an full active eavesdropper with r=R, SOP in jamming with $Pj\leq Ps is always less than eavesdropping with Ps
%eavesdropping is more severe than in jamming 
% We observe that a jammer with $P_{\mathcal{J}}$\leq $P_{\mathcal{S}}$ can never yield  SOP $\hat{p}^\mathcal{J}_{o}$ greater than $\hat{p}^\mathcal{E}_{o}$ posed by  an strong eavesdropper with $r=R$. Hence under these conditions attacker prefers to do eavesdropping as compared to jamming with aim of posing maximum secrecy outage to user. 
%%%%%%%%%%%

{We observe that under the same channel conditions for eavesdropping and jamming link, $\mathcal{A}$ in jamming mode with $P_{\mathcal{J}} < P_{\mathcal{S}}$  yields  SOP $\hat{p}^\mathcal{J}_{o}$ less than $\hat{p}^\mathcal{E}_{o}$   in eavesdropping mode with $r=R$. Hence, under these conditions, $\mathcal{A}$ prefers  eavesdropping relative to jamming  to cause the maximum secrecy outage to the legitimate nodes. However, for $r<R$, jammer becomes more harmful than eavesdropper beyond a definite $ C_{st}$.}
\begin{figure}[!t]
\centering
{{\includegraphics[width=2.5in]{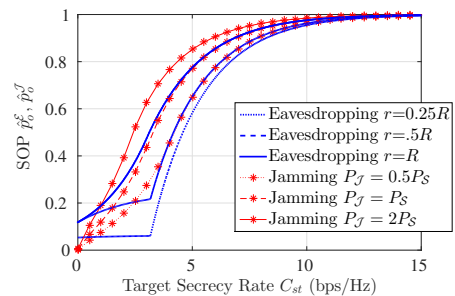}} }%\vspace{-5mm}
%{{\includegraphics[width=2.5in]{Figures/EavD_RFig7.eps}} }%\vspace{-5mm}

\caption{\footnotesize \textit{ Variation in SOP  under eavesdropping  and jamming with  $R =100$  m, $D=50$ m.}} 
    \label{fig: Eaves_jam}
    \vspace{-2mm}
\end{figure}
%%%%%%%%%%
%%%%%%%%%%%%%%%%%%%%%%%%%%%%%%%%%%%%%%%%%%%
This condition is true for both cases: with $P_{\mathcal{J}}<P_{\mathcal{S}}$ and $P_{\mathcal{J}}\geq P_{\mathcal{S}}$.
% relative to   eavesdropper with $r= 0.5 R$,  jammer with $P_{\mathcal{J}}=0.5 P_{\mathcal{S}}$ poses more secrecy outage  to user than $ C_{st}>2.5$ bps/Hz. Jammer with $ P_{\mathcal{J}} =2P_{\mathcal{S}} $ is more harmful than eavesdropper for $C_{st}>1$ bps/Hz. 
 At a very low value of $C_{st}$ (near to zero), SOP is higher under eavesdropping  because there may be a probability of zero SC which includes the probability of negative secrecy  capacity as defined by \eqref{eq: C_E}, whereas there is always positive SC under jamming.
  \begin{remark}\label{rem:severity}
 \textit{ By maintaining $D \leq D_o$, SOP $\hat{p}^\mathcal{E}_{o}$ in eavesdropping can be limited to $\alpha=r^2/R^2$. Furthermore, eavesdropping has more severe impact on secrecy performance  than jamming for the same channel conditions for eavesdropping and jamming link unless $r<R$ or $P_{\mathcal{S}}<P_{\mathcal{J}}$. }  
  \end{remark}
 
\section{Concluding Remarks
}\label{sec:conclusion}
 This paper  proposes a novel QoS-aware stochastic system model and investigates the impact of random inter-node distances  under eavesdropping as well as jamming. Development of  ratio distribution of legitimate to attacker link SNR for the proposed system in  closed-form makes an important and significant contribution in the knowledge domain. This has enabled us to derive novel closed-form expressions for SOP. Analytical expressions have been validated against simulation results with almost perfect match. New insights are drawn for designing of system parameters in order to achieve the desired SOP. The proposed model has been shown to design maximum $\mathcal{S}$-$\mathcal{U}$ separation for achieving desired secrecy performance. From attacker point of view, it is shown through numerical results that an  optimal value of eavesdropping range can be designed to cause maximum secrecy outage. It is also shown that no gain is achieved by $\mathcal{A}$ beyond a threshold value of eavesdropping range for given conditions. In the case of jamming, $\mathcal{A}$ gets a diminishing returns beyond a threshold value of jamming power. Also, the conditions are discussed under which there is a linear increment in the maximum allowed separation between legitimate nodes  with increased deployment range for an acceptable SOP. Finally, we compare the severity of eavesdropping and jamming in terms of secrecy outage. It is observed that  for the same set of channel conditions for  $\mathcal{S}$-$\mathcal{A}$ and $\mathcal{A}$-$\mathcal{U}$ links, $\mathcal{A}$ must have more transmit power  than $\mathcal{S}$ to pose more secrecy outage in jamming than eavesdropping provided that it is capable to eavesdrop in entire deployment range.  It is, therefore, established that the proposed analytical model offers a reference framework for inferring deep insights into functioning of eavesdropping and jamming. It can also be leveraged as an analytical design tool for determining system parameters to achieve the required secrecy performance.
%\bibliographystyle{imaiai}
%\bibliographystyle{iet_fixed}
 %\bibliographystyle{IEEEtran}
%\bibliography{references_twc} 
%\bibliography{paper_bh} 

% Generated by IEEEtran.bst, version: 1.14 (2015/08/26)

\end{document}